\documentclass[sigplan,twocolumn]{acmart}
\renewcommand\footnotetextcopyrightpermission[1]{}
\usepackage{tikz}
\usepackage{amsmath}

\usepackage{filecontents}

\usepackage{xspace}
\usepackage{multirow}
\usepackage{booktabs} 
\usepackage{tabularx}
\usepackage{hyperref}
\usepackage{enumerate}
\usepackage{enumitem}
\usepackage{algorithm}
\usepackage{algorithmic}
\usepackage[most]{tcolorbox}
\usepackage{makecell}
\usepackage{balance}
\usepackage{caption}
\usepackage[skip=0cm,list=true]{subcaption}
\usepackage[countmax]{subfloat}
\usepackage{xurl}
\usepackage{wrapfig}
\usepackage{graphicx}

\usepackage{amsfonts}
\usepackage{amsthm}
\newtheorem{definition}{Definition}

\usepackage[utf8]{inputenc}
\usepackage{xcolor}

\newcommand{\sys}{\textsc{Citadel++}\xspace}
\newcommand{\veps}{\varepsilon}
\newcommand{\norm}[1]{\Vert #1 \Vert}
\newcommand{\ee}{{\rm e}\hspace{1pt}}
\newcommand{\dd}{\hspace{1pt}{\rm d}\hspace{0.5pt}}

\newtheorem{thm}{Theorem}

\begin{document}

\date{}

\title[]{\LARGE Protecting Confidentiality, Privacy and Integrity in Collaborative Learning}

\author{
{\rm Dong Chen$^{\dag,*}$ \quad Alice Dethise$^{\S}$ \quad Istemi Ekin Akkus$^{\S}$ \quad Ivica Rimac$^{\S}$}\\
{\rm Antti Koskela$^{\S}$ \quad Marco Canini$^{\ddag}$ \quad Wei Wang$^{\dag}$ \quad Ruichuan Chen$^{\S}$}\\
$^{\S}$Nokia Bell Labs \qquad $^{\dag}$HKUST \qquad $^{\ddag}$KAUST
} %

\begin{abstract}

A collaboration between dataset owners and model owners is needed for effective machine learning (ML) training.  
During this collaboration, however, dataset owners and model owners want to protect the \emph{confidentiality} of their respective assets (i.e., datasets, models and training code), with the dataset owners also caring about the \emph{privacy} of individual users whose data is in the datasets.
Existing solutions provide limited confidentiality for models and training code, suffer from data privacy issues, or incur unacceptable performance.

We present \sys, a collaborative ML training system to \emph{simultaneously} protect the confidentiality of datasets, models and training code as well as the privacy of individual user data.
Specifically, \sys takes a \emph{chained} approach.
First, we employ the Virtual Machine-level \emph{Trusted Execution Environments} (TEEs) as the training runtime to provide confidentiality protections for datasets, model and training code.  Second, we 
create a
\emph{privacy barrier} that improves upon differential privacy techniques to safeguard the privacy of individual user data while maintaining model utility. Third, we construct and optimize a Virtual Machine Privilege Level (VMPL)-enforced fine-grained \emph{sandbox} to regulate the data flow in the system to ensure that our privacy barrier is always applied even when the model and training code may be maliciously designed.
Finally, we employ specific TEE capabilities to ensure the \emph{integrity}
of our privacy barrier and sandboxing mechanisms for their correct execution.

To our knowledge, \sys is the \emph{first} collaborative training system providing strong confidentiality and privacy protections for dataset owners and model owners simultaneously.
We analyze \sys' security properties and implement a full prototype. 
Our evaluation shows that \sys provides strong guarantees with modest overheads.
Importantly, \sys enables new, previously-infeasible ML training workloads in a mutually distrustful setting.

\end{abstract}

\maketitle
\begingroup
\renewcommand\thefootnote{\fnsymbol{footnote}}
\footnotetext[1]{Work partly done while Dong Chen was an intern at Nokia Bell Labs.}
\endgroup
\pagestyle{plain}
\thispagestyle{plain}

\section{Introduction}
\label{sec:intro}

High-quality machine learning (ML) training requires a well-designed model and representative datasets. 
Unfortunately, \emph{dataset owners} often own siloed user data but may not have extensive ML 
expertise.
For instance, hospitals have patient data but 
may lack expertise in advanced ML techniques to utilize these datasets.
On the other hand, \emph{model owners} have ML expertise, but they need datasets from diverse sources to train their models.
For instance, a biotech company with a new drug design model may need data 
from a variety of hospitals to generalize the model.
As a result, the collaboration between dataset owners and model owners is required.

Although there are various attempts at enabling this  collaboration~\cite{genes2021defs, genes2022data, aws-data-exchange, databricks-marketplace, datarade.ai},
dataset owners and model owners do not necessarily trust each other.
Dataset owners care about their datasets. For instance, they may want to monetize these datasets via marketplaces~\cite{aws-data-exchange, snowflake-marketplace, databricks-marketplace}. In addition, these datasets may contain sensitive user data, raising privacy concerns
and being subject to various privacy regulations (e.g., EU AI Act~\cite{eu-ai-act}, GDPR~\cite{gdpr}, CCPA~\cite{ccpa}, HIPAA~\cite{hipaa}).
Any leakage of such sensitive data can lead to severe fines, business impact and reputation loss.

Similarly, model owners care about their models which are becoming innovative. 
Model owners spend significant time and resources designing the model structure and the way 
it can be trained efficiently; thus, models and the training code are often considered as intellectual properties.
Although there are open-sourced  models 
in domains such as natural language processing~\cite{llama, deepseek, qwen} and image generation~\cite{diffusion}, there are many confidential models that have been employed (and will likely be employed in the future) in many industries, such as drug design models in pharmaceuticals and stock price prediction models in finance.

As a result, there is a dilemma.  Even if the collaboration between dataset owners and model owners is  beneficial for both sides, neither side is willing to openly share their respective assets.
Existing attempts to solve this dilemma either address it partially or impose unacceptable performance overheads.
For instance, Federated Learning (FL)~\cite{konevcny2016federated, mcmahan2017communication, li2020federated, kairouz2021advances, mcmahan2016federated} enables dataset owners to train a model over unshared local datasets and synchronize with a server that aggregates all the local model updates to generate an updated model.  
FL, however, has been shown to be vulnerable to data privacy attacks such as data reconstruction~\cite{zhu2019deep, yin2021see, geiping2020inverting, hitaj2017deep} %
and membership inference~\cite{shokri2017membership, nasr2019comprehensive, carlini2022membership}.
In addition, FL lacks model confidentiality because the model is made available at the dataset owners.
In parallel, the community has been exploring collaborative training based on cryptographic approaches, such as homomorphic encryption~\cite{bost2014machine, brickell2007privacy, graepel2012ml, li2017multi, aono2017privacy, stevens2022fldp} or secure multi-party computation~\cite{nandakumar2019towards, tian2022sphinx, hesamifard2018privacy, sav2020poseidon, mohassel2017secureml, agrawal2019quotient, mohassel2018aby3, wagh2019securenn, patra2020blaze, chaudhari2019trident, koti2021tetrad, ng2023sok, wagh2020falcon, liu2024pencil}.  
However, these approaches impose significant computation and communication overheads~\cite{liu2021machine}.

Recently, confidential computing has become an emerging paradigm that leverages hardware-based Trusted Execution Environments (TEEs) to provide assurances for confidentiality and integrity when executing code on data within a TEE~\cite{confidential2022technical, confidential2020confidential}.
This paradigm has gained significant traction among CPU and GPU manufacturers~\cite{intel-sgx-1, mckeen2013innovative, amd-sev, intel-tdx, arm-cca, apple-pcc, nvidia-cc}, and has become an essential offering of all major cloud providers~\cite{apple-pcc, azure-cc, google-cc, aws-nitro}.
TEEs are now regarded as key to safeguarding datasets and models in ML processes~\cite{aws-clean-rooms, snowflake-clean-rooms, huba2022papaya, impact-of-cc, future-cc1}.
Indeed, a large body of research has leveraged TEEs for privacy-preserving ML computations~\cite{ohrimenko2016oblivious, tramer2018slalom, hunt2018chiron, lee2019occlumency, mo2021ppfl, quoc2021secfl, mo2020darknetz, hynes2018efficient, quoc2020securetf, hunt2020telekine, zhang2021citadel, hua2022guardnn, vaswani2023confidential, mai2023honeycomb}.
However, 
these approaches
require the model and its training code to be apriori \emph{inspected} to ensure they are not maliciously designed to leak private data while running inside TEEs, breaking the needed confidentiality for the model and training code.

In this paper, we propose \sys, a collaborative training system to \emph{simultaneously} protect the confidentiality of datasets, models and training code as well as the privacy of individual users whose data is present in those datasets.  \sys achieves these guarantees even in an environment where the model and its training code may be malicious and cannot be inspected due to confidentiality requirements. 

First, we adopt an FL-style programming paradigm
where
the model training code is split into two parts: 1) the \emph{data handling code} operates on datasets to compute gradients, and 2) the \emph{model updating code} aggregates gradients to update the model.  We employ Virtual Machine (VM)-level
TEEs (e.g., AMD SEV-SNP~\cite{amd-sev}) as runtime to provide a training process with 
\textbf{confidentiality} for datasets, model and training code, as well as to enable the usage of confidential GPUs~\cite{nvidia-cc}. 

Second, to ensure strong privacy for individual user data, \sys 
creates a \textbf{privacy barrier} between the data handling code and the model updating code.  It improves upon DP-SGD related techniques~\cite{abadi2016deep, andrew2021differentially, bonawitz2017practical, kairouz2021distributed, golatkar2022mixed} with three additional mechanisms: 1) attack-resistant single-item processing, 2) practical dynamic gradient clipping, and 3) efficient differentially-private masking.
Our privacy barrier ensures that any data (e.g., gradients) sent from data handling code to model updating code is differentially-private, so that
little information about any individual user can be gained through such data, 
even in the presence of 
collusion among dataset owners and the model owner, without sacrificing utility.

Third, we construct a \textbf{sandbox}, based on Virtual Machine Privilege Levels (VMPLs)~\cite{amd-sev,vmpl}, to regulate the data flow in \sys.
This sandbox is needed because providing confidentiality for the model and training code within
TEEs 
also prevents anyone from inspecting whether they are maliciously designed 
to bypass our privacy barrier.
Our sandbox isolates the trusted/inspectable \sys service code and its sandbox monitor 
from the untrusted/uninspectable model and training code (provided by the model owner), and enforces 
that the data exchanged between the data handling and model updating codes flows \emph{only} through our privacy barrier.
Specifically, each sandbox processes only one data item, without any shared state or communication between sandboxes, for required privacy.
To reduce overheads from this requirement, 
we optimize our sandboxing mechanism.

Finally, we guarantee the \textbf{integrity} of our privacy barrier and sandboxing mechanisms as well as their correct execution.  In particular, we build upon specific attestation and image-integrity schemes from Confidential Containers (CoCo)~\cite{coco} to ensure that the \sys service code (which houses our privacy barrier and sandboxing mechanisms) is deployed correctly and protected from tampering at runtime.

In summary, \sys employs a chained design: 1) VM-level TEEs provide the confidentiality for datasets, models and training code, 2) the privacy barrier safeguards the privacy of individual user data with high model utility, 3) the VMPL-enforced fine-grained sandboxing ensures the privacy barrier is always applied even if the model and training code may be malicious, and 4) the integrity mechanism ensures the correct execution of our privacy barrier and sandboxing.

We analyzed and implemented \sys.
With our privacy barrier, training approaches non-private accuracy 
with negligible overhead
(\S\ref{sec:evaluation-dp}).
Our optimizations speed up sandboxing by 1.27--2.98$\times$ over direct VMPL isolation (\S\ref{sec:evaluation-sys}).
\sys scales sub-linearly with batches up to 8192 and with up to 64 dataset owners, can support confidential GPUs (\S\ref{sec:evaluation:scalability}), and
is up to 2 orders of magnitude 
faster than the crypto-based system~\cite{liu2024pencil}
and 32.6-129.7$\%$ slower than other TEE-based system~\cite{zhang2021citadel} that provides weaker guarantees (\S\ref{sec:evaluation-comparison}).

To our knowledge, \sys is the \emph{first} collaborative training system simultaneously protecting the confidentiality of datasets, models and training code, as well as the privacy of individual user data, even if the model and training code
cannot be inspected due to confidentiality. With these protections, \sys solves the tension between dataset owners and model owners, and enables new, previously-infeasible training workloads in a mutually distrustful 
environment.

\section{Actors, Assumptions and Goals}
\label{sec:assumptions}

\subsection{System Actors}
\label{sec:actors}

\noindent\textbf{Dataset owners} are entities that possess datasets that may
contain sensitive or private data. 
They care about not only the confidentiality of their datasets, but also the privacy of individuals whose data is present in those datasets.  In this work, we target a cross-silo scenario where dataset owners are typically organizations, so that there can be a few tens of dataset owners in a collaborative training session. 

\noindent\textbf{Model owners} are entities 
that create models and the associated training code.
Model owners may own some datasets to develop and test their models and training code;
however, they do not possess enough and diverse datasets to generalize their models.
They want to collaborate with dataset owners to access such datasets,
but 
they care about the confidentiality of their models and training code.

\noindent\textbf{\sys provider} is the entity that deploys and operates 
the \sys system to enable the collaborative training among dataset owners and model owners.
It interfaces with the cloud provider for the necessary infrastructure.

\noindent\textbf{Cloud provider} is the entity that provides 
the underlying compute, network and storage infrastructures with their respective services.
The compute infrastructure is equipped with VM-level TEEs (e.g., AMD SEV-SNP~\cite{amd-sev}), which has been offered by the major cloud providers%
~\cite{aws-snp, azure-cc, google-cc}.

\subsection{Assumptions on Trust and Reliability}
\label{sec:threat}

\noindent\textbf{Trust Assumption.} \sys adopts the standard trust model on VM-level TEEs regarding the confidentiality and integrity protections~\cite{sev-white-paper, cheng_tdxsurvey_2023, tdx_google_security_review, linux_coco_threat_model, confidential2022technical}.  
We trust the TEE hardware and firmware. 
The attacks on TEEs, such as covert and side channels~\cite{li2022systematic, lee2020off, murdock2020plundervolt, schluter2024wesee}, physical attacks~\cite{demeulemeester24-badram, chen2021voltpillager} and rollback attacks~\cite{trikalinou2017taking}, are out of scope.

We consider an adversary can access privileged system stacks with the goal of learning confidential information:
The cloud provider may manipulate the host OS kernel and process resources;
the \sys provider may manipulate \sys's system components.
In addition, each actor (as defined in \S\ref{sec:actors}) may try to learn the confidential datasets, models and training code, either individually or by colluding with others.
The model and training code provided by a model owner may be maliciously designed
to break the confidentiality and privacy of datasets.

As in standard trust model of VM-level TEEs, the guest kernel is trusted.
\sys service code, which implements our privacy (\S\ref{sec:design-dp}) and sandboxing mechanisms (\S\ref{sec:sandbox}), is inspectable and trusted.
The deployment and execution of this service code are protected by the integrity mechanism (\S\ref{sec:integrity}).  

\noindent\textbf{Reliability Assumption.}  
We assume reliable compute, network and storage infrastructures with no denial-of-service.

\begin{figure*}[t]
	\centering
    \includegraphics[width=0.78\linewidth]{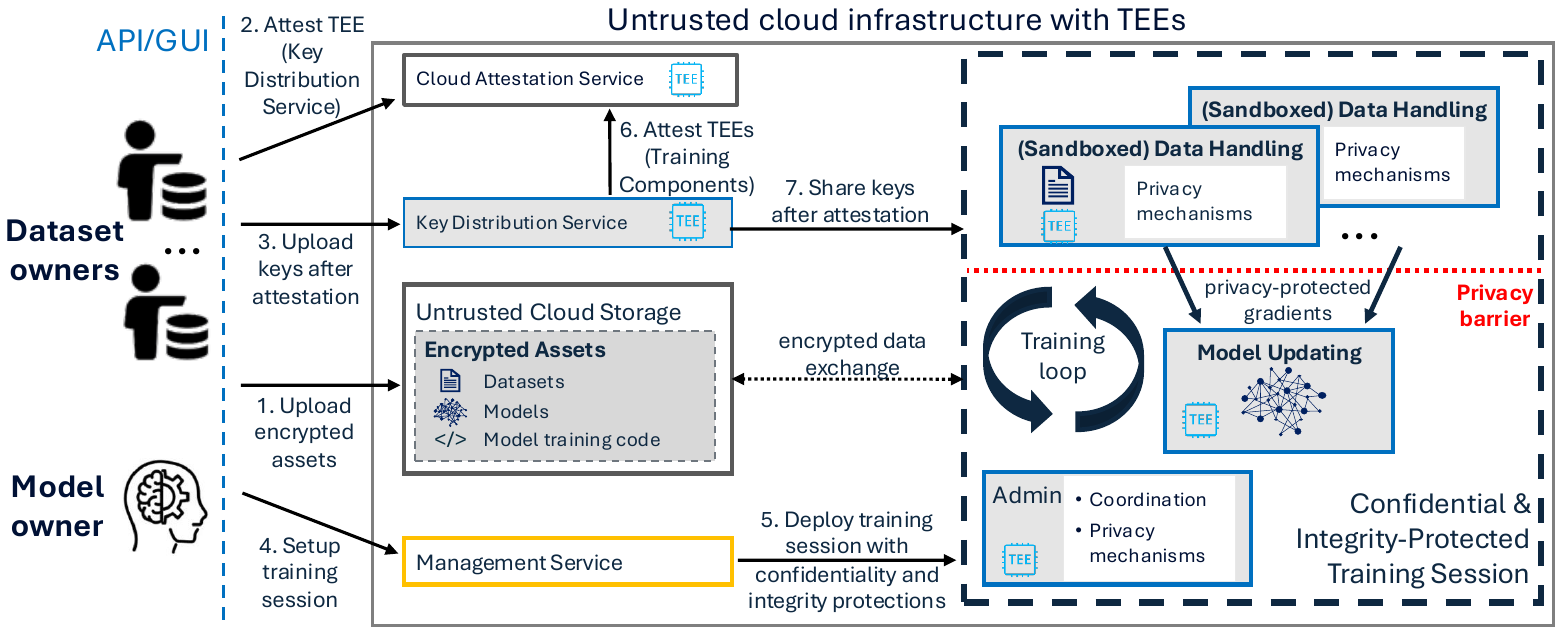}
	\caption{The overview of \sys' components, mechanisms and workflow. 
    }
	\label{fig:overview}
\end{figure*}

\subsection{System Goals}
\label{sec:goals}

\noindent\textbf{Dataset/Model/Code Confidentiality:} 
Datasets may contain sensitive information. Similarly, ML models, including their structures, parameters and the associated training code, constitute intellectual properties.
The confidentiality of these assets should be protected.
This protection should apply during storage (i.e., at rest), during transfer (i.e., in transit) and during computation (i.e., in use).

\noindent\textbf{Data Privacy:} A dataset may be from one organization (e.g., operational data from a factory), but may also contain data from multiple individual users (e.g., medical records of hospital patients).
Data of each individual user should be protected,
so that
little information about any individual user should be gained or inferred in \sys. 
Note that, this data privacy protection goes beyond what the aforementioned dataset confidentiality provides for a dataset as a whole.

\noindent\textbf{Code Integrity:} The code running in \sys should be integrity-protected.
This protection applies to both the training code (from the model owner) and the \sys service code.
The service code includes our protection mechanisms (\S\ref{sec:design-dp}, \S\ref{sec:sandbox} and \S\ref{sec:integrity}), and is crucial for achieving our confidentiality and privacy goals for datasets, models and training code.

\noindent\textbf{System Performance:} 
\sys should support efficient confidential collaborative training, while preserving model accuracy and meeting the confidentiality, privacy and integrity goals.
Training should scale with the number of participants and batch size, with acceptable overheads from the protection mechanisms to achieve our goals simultaneously.
 
\section{\sys Overview}
\label{sec:design}

This section gives an overview of the \sys design.

\subsection{Background: TEE Capabilities and VMPLs}

Within a processor, a TEE is a secure area
isolated from the untrusted environment, providing confidentiality and integrity assurances for the code and data inside~\cite{confidential2022technical, confidential2020confidential}.
Remote attestation~\cite{costan2016intel, haldar2004semantic} enables users to authenticate the TEE hardware, verify its state, and check if the expected code is running via its cryptographic measurement.
A TEE can protect an enclave in a process, as in Intel
SGX~\cite{intel-sgx-1}, or an entire VM, as in AMD SEV-SNP~\cite{amd-sev} and Intel
TDX~\cite{intel-tdx}.
VM-level TEEs protect guest memory from the host
OS and hypervisor, allowing a Confidential VM (CVM) to run its own guest OS and
applications.
A CVM also offers a layered privilege model for its processes.
For instance, AMD SEV-SNP provides four hardware-enforced Virtual Machine Privilege Levels (VMPLs),
with higher levels (e.g., VMPL0) performing the access control of lower levels (e.g., VMPL3)
~\cite{amd-sev,vmpl}.
Intel TDX offers TD-partitioning via nested virtualization
~\cite{td-partitioning}.

\subsection{System Components and Mechanisms}
\label{sec:design:overview}

Figure~\ref{fig:overview} shows an overview of \sys.
The system components enable a collaborative training session between the dataset owners and model owner, while providing confidentiality, privacy and integrity protections.
These components run in an untrusted cloud infrastructure equipped with TEEs and an untrusted \textbf{storage service}.
Therefore, all training assets
stored in the cloud storage before, during and after the training session must be encrypted (i.e., encryption-at-rest).

To make these encrypted assets usable in the confidential and integrity-protected training session with TEEs (i.e., encryption-in-use),
the corresponding keys have to be supplied by their respective owners to the relevant TEE components.
For ease-of-use, \sys manages these keys inside its \textbf{key distribution service} (KDS) running in a TEE.
The KDS is responsible for 1) storing keys of various assets along with the agreed training session configurations, and 2) sharing those keys with the relevant TEE components at training start.
Before sharing the keys, the KDS remotely attests the relevant TEE components via the trusted \textbf{cloud attestation service} provided by the 
cloud provider~\cite{azure-attestation},
and only shares keys if these components run the intended service code and possess the agreed training session configurations.
These components in their respective TEEs can be remotely attested~\cite{costan2016intel, haldar2004semantic},
because the service code of all \sys components, including the KDS, is open-sourced.

We design a \textbf{management service} that sets up training sessions and tracks various metadata 
(e.g., participants, whether all assets are available for training to start, and the training progress).
The management service is untrusted. 
The KDS ensures the training session configuration supplied by the management service is consistent with its metadata.

In \sys, the collaborative training follows an FL-style paradigm.  
Specifically, the model owner's training code is split into two parts: 
data handling code and model updating code.
Note that, for confidentiality reasons, the training code including its both parts cannot be inspected, thus untrusted.
For each dataset owner, a \textbf{data handling component} runs in a TEE 
with the \sys service code, which executes the (untrusted) data handling code
to confidentially process this dataset owner's data and generate the local model updates.
For the model owner, a \textbf{model updating component} runs in a TEE with the \sys service code,
which executes the (untrusted) model updating code to confidentially aggregate all local model updates and update the global model.
There is another TEE component, the \textbf{admin component},
which also runs the \sys service code to coordinate training iterations among the data handling and model updating components, and set up our privacy barrier mechanism.

All 
\sys \textbf{service code} is open for inspection and trusted.
This service code includes:
1) the \textbf{privacy barrier} that ensures the local model updates exchanged between data handling and model updating components do not leak private information (\S\ref{sec:design-dp}), 
2) the \textbf{sandboxing mechanism} that isolates the potentially malicious model and training code,
and ensures our privacy barrier cannot be bypassed (\S\ref{sec:sandbox}), 
and
3) various \textbf{procedures} that handle the orchestration of the training loop and the common utility actions such as encrypting/decrypting data among TEE components and the untrusted cloud storage (i.e., encryption-in-transit).

The \sys service code is key to protecting the confidentiality and privacy of the datasets, model and training code as well as the correct training; thus, the integrity of this service code must be protected at deployment and runtime.
TEE memory protection alone is insufficient to establish this property.
We therefore employ extra \textbf{integrity mechanisms} to protect the \sys service code and the training code it manages from tampering, and ensure all our protection mechanisms are deployed and executed as intended (\S\ref{sec:integrity}).

\subsection{\sys Workflow}
\label{sec:design:workflow}

After agreeing to join a collaborative training session, 
the dataset owners and model owner encrypt their respective assets (i.e., datasets, model and training code) with the keys of their choosing.
They upload their encrypted assets to the untrusted cloud storage (\textbf{Step 1} in Figure~\ref{fig:overview}).
For the training to proceed, these encrypted assets need to be decrypted inside the TEEs running the \sys training components, thus requiring keys to be shared with the respective components.
To do so, the dataset owners and model owner remotely attest the key distribution service (KDS) via the trusted cloud attestation service, and 
upload their keys and the agreed training session configuration to KDS (\textbf{Steps 2 and 3}).

Next, the management service starts the confidential and integrity-protected training session with the training session configuration that refers to the assets and the privacy parameters to be used, and then deploys the associated training components within TEEs (\textbf{Steps 4 and 5}):
A separate data handling component for each dataset owner, a model updating component, and an admin component.  Note that, a dataset owner with a large dataset can be allocated to multiple concurrent data handling components, each of which processes a disjoint part of the dataset. Similarly, a large training batch can also be handled with this scheme.

After these components are deployed, the \sys service codes in the data handling and model updating components register with the admin component, 
which configures them with the relevant assets and privacy parameters:
Each data handling component receives the corresponding dataset ID and
the data handling code ID.
The model updating component receives the initial model ID, the model updating code ID, and the test dataset ID (for assessing the trained model).
The service code in each component retrieves its still encrypted assets from the untrusted cloud storage
and requests their corresponding keys from KDS.
The KDS remotely attests the training components via the cloud attestation service, ensures the key metadata matches the training session configuration, and then shares the relevant keys (\textbf{Steps 6 and 7}).
Once shared, the KDS discards the keys, so that a model owner cannot run multiple training sessions with the same datasets without involving the dataset owners.

Finally, the \textbf{training loop} starts.
At the beginning of each iteration, the admin component 
generates and distributes differentially-private masks (i.e., DP-masks, \S\ref{sec:design-dp}) to the data handling components.
Each data handling component loads the current iteration's dataset items and model, generates the gradients (i.e., model updates)\footnote{In this paper, we use model updates and gradients interchangeably.} via the data handling code, and applies the DP-masks received from the admin component.
Afterwards, it encrypts the DP-masked gradients and sends them to the model updating component.
Upon receiving the DP-masked gradients, the model updating component decrypts and aggregates them with the current model to obtain the updated model for the next iteration
via the model updating code.
The model updating code may assess the model accuracy with a confidential test dataset
before signaling the admin component
whether to stop training.
The admin component also checks whether the training session's allocated privacy budget is fully spent.
If any of these
conditions is true,
the training is stopped.

Throughout the training session, our sandboxing mechanism (\S\ref{sec:sandbox}) ensures our privacy barrier is always applied even when the model and training code (i.e., both the data handling and model updating codes) may be malicious (\S\ref{sec:design-dp}).
Our design also employs specific attestation and image-integrity mechanisms to authenticate the \sys service code that implements our privacy barrier and sandboxing mechanisms, and protect its execution environment from tampering (\S\ref{sec:integrity}).

\section{Privacy Barrier}
\label{sec:design-dp}

Here, we describe our privacy barrier that protects the privacy of individual user data while retaining model utility.
It improves upon DP-SGD techniques~\cite{abadi2016deep, andrew2021differentially, bonawitz2017practical, kairouz2021distributed, golatkar2022mixed} with three mechanisms: single-item processing for \textbf{attack resistance}, dynamic gradient clipping for \textbf{practicality}, and differentially-private masking for \textbf{efficiency}. 

\subsection{Background: DP-SGD}
\label{sec:dpsgd}

DP-SGD~\cite{abadi2016deep} is a standard way to ensure DP guarantees for an ML training process.
The guarantees are achieved by perturbing the generated gradients (with bounded $L2$-norm) with Gaussian noise.  
A randomized mechanism $\mathcal{M}$ satisfies $(\veps,\delta)$-DP if for any two neighboring datasets $D$ and $D'$ (differing in at most one data item) and for any subset $S$ of possible outputs, it holds: $\Pr[\mathcal{M}(D) \in S] \leq \ee^\varepsilon \Pr[\mathcal{M}(D') \in S] + \delta$.  The two privacy parameters, $\veps$ and $\delta$, control the trade-off between the model utility and privacy guarantee.

Applying DP-SGD to an FL scenario often leverages a secure aggregation scheme~\cite{bonawitz2017practical, kairouz2021distributed}.
Secure aggregation ensures that each dataset owner adds a small amount of noise to its generated gradients, while the aggregate noise from all dataset owners is the minimally required DP noise.
In addition, secure aggregation ensures that the server in FL learns only the noisy aggregate gradient but not individual (insufficiently-noised) gradients.  
However, secure aggregation typically incurs significant cryptographic and communication overheads, making it inefficient for realistic ML training.
To address the performance issue, Papaya~\cite{huba2022papaya} and Citadel~\cite{zhang2021citadel} utilize TEE capabilities for secure aggregation. 

\subsection{Single-Item Processing}
\label{sec:privacy_barrier:item}
Standard DP-SGD assumes the gradient computation is executed faithfully.
In \sys, however, the potentially malicious data handling code (from the model owner)
could manipulate the gradient computation to mitigate the DP protections. 
Imagine the data handling code receives data items in batches of size $B$.
When it sees a sensitive data item, 
it drops the other $B-1$ items and replaces them with copies of the sensitive item.
Afterwards, the normal gradient computation and DP-SGD apply. 
In the end, the duplicated data item is overly represented, making its presence easier to detect; thus, breaking the DP guarantee for individual data items.

We improve DP-SGD and make it \textbf{resistant} to such gradient manipulation attacks.
We give the untrusted data handling code a single data item at a time.
For performance, we run multiple parallel sandboxed processes (\S\ref{sec:sandbox}), each executing the data handling code, within a data handling component.
Specifically, with a batch size of $B$, 
the trusted \sys service code of a data handling component passes each data item \emph{individually} to a separate process,
collects the computed gradient, and assembles the collected gradients.
This mechanism does not affect the DP guarantee (Appendix \ref{sec:security_analysis_formal_dp}).

Our sandboxing mechanism (\S\ref{sec:sandbox:isolation}) ensures that the untrusted data handling code cannot communicate and maintain state across separate processes.
As a result, each data item can be present at most once\footnote{Data items can be dropped, which does not invalidate the DP guarantees.} in the 
computed gradients.

\subsection{Dynamic Gradient Clipping}
\label{sec:privacy_barrier:clipping}
Standard DP-SGD applies a bounded norm for each gradient via clipping~\cite{abadi2016deep}, which scales down any gradient whose norm exceeds a \emph{fixed} clipping bound $C$.
Selecting the optimal fixed clipping bound is a non-trivial process.
A larger $C$ retains more gradient information, but requires more noise 
for the same DP guarantee;
a smaller $C$ leads to clipping more gradients.
Both can degrade model utility, requiring an adaptive choice of the clipping bound $C$. 
Prior work~\cite{andrew2021differentially, golatkar2022mixed} explores adaptive clipping but is lengthy and quite involved.

\sys adjusts the clipping bound \emph{dynamically} in each iteration, by clipping approximately $r$ percentage of the gradients. 
Specifically, within TEEs, each data handling component computes the norms of its generated gradients in an iteration and constructs a count histogram of these norms using pre-defined bins $\{B_i\}_{i=1}^{k}$.
The service code in the admin component receives these histograms from all data handling components, sums them up to obtain a histogram estimate of the norms across all gradients, and adds Gaussian noise to the counts with the noise variance $\sigma_g^2$.
The norm matching the $r$-th percentile is selected and returned to data handling components as the clipping bound $C$ for that iteration.

Our dynamic gradient clipping keeps gradient norms and histograms confidential inside TEEs with sandboxing (\S\ref{sec:sandbox}) and integrity protections (\S\ref{sec:integrity}),
avoiding potential privacy leakage.
It is \textbf{practical} because it does not require a non-trivial tuning or adaptation of a fixed clipping bound~\cite{andrew2021differentially, golatkar2022mixed}.
Finally, it is proven to be differentially private (Appendix \ref{sec:security_analysis_formal_dp}). 

\begin{figure}[t]
	\centering
        \includegraphics[clip=true, trim={30 220 100 100}, width=\linewidth]{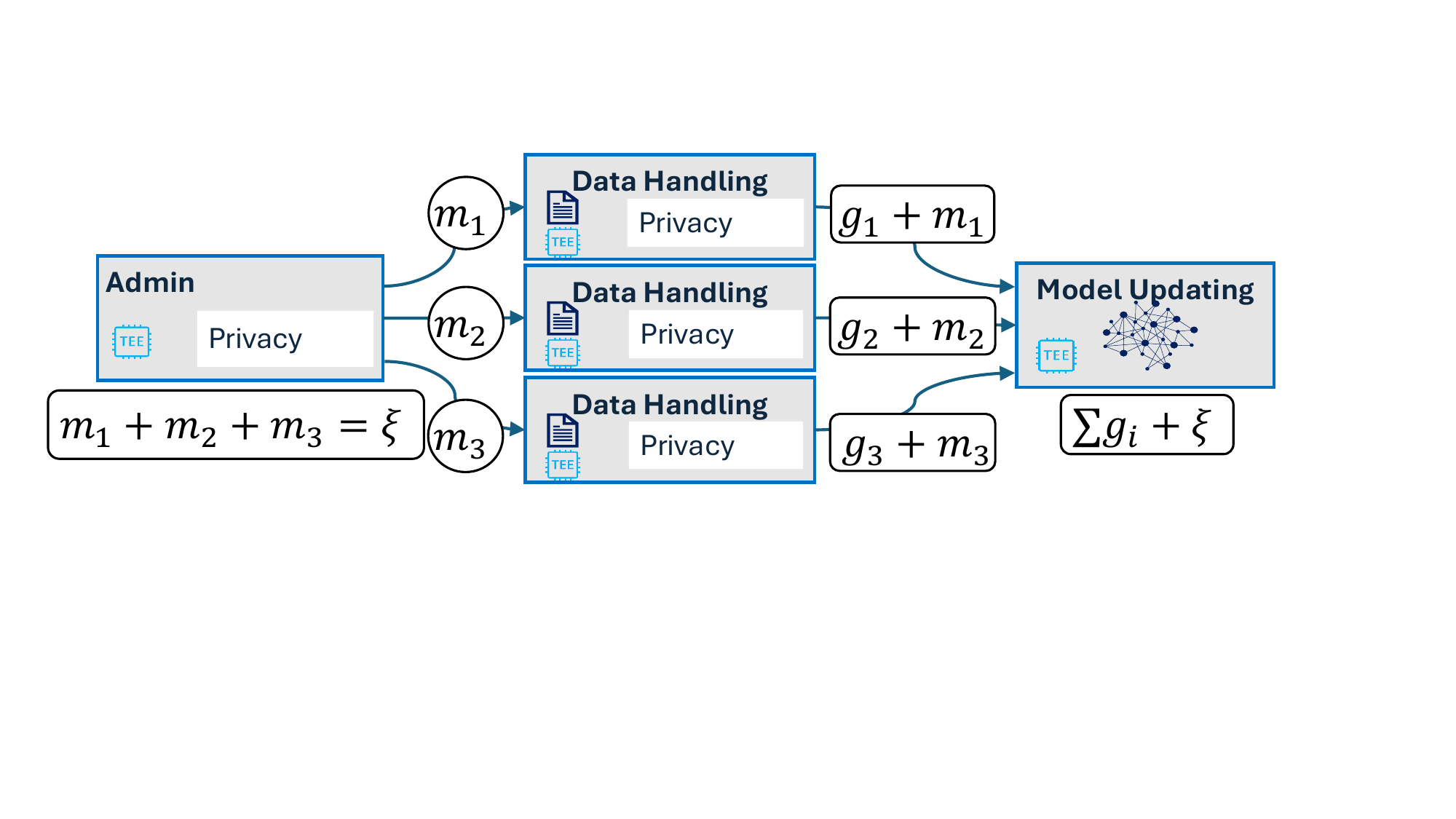}
	\caption{Differentially-private masking.}
	\label{fig:masked_gradients}
\end{figure}

\subsection{Differentially-Private Masking}
\label{sec:privacy_barrier:dpmasking}
The potentially malicious model updating code 
could utilize the received gradients to infer sensitive information about the training datasets~\cite{zhu2019deep, yin2021see, geiping2020inverting, hitaj2017deep, shokri2017membership, nasr2019comprehensive, carlini2022membership}.
Applying a heavy secure aggregation scheme with DP guarantees is a typical solution (\S\ref{sec:dpsgd}).
To accelerate secure aggregation, recent work~\cite{huba2022papaya, zhang2021citadel} leverages TEE capabilities~\cite{confidential2022technical, confidential2020confidential} to simplify its execution, but often without DP guarantees.
In contrast, 
we design a DP-masking mechanism
that provides both \textbf{efficiency} and \textbf{DP guarantees}.

Specifically, the service code of the admin component generates random DP-masks such that their sum provides DP guarantees, i.e., $\sum_{i=1}^{n} m_i = $$\mathcal{N}(0, \sigma^2 C^2 \mathbf{I}) = \xi$, where $m_i$ is the DP-mask sent to the $i$-th data handling component, $\sigma$ is the noise scale parameter, and $C$ is the gradient clipping bound (Figure~\ref{fig:masked_gradients}).
The admin component then distributes these masks to the respective data handling components.
The service code of each data handling component applies its received DP-mask $m_i$ to its local gradient $g_i$ before sending the DP-masked gradient $g_i + m_i$ to the model updating component,
where each such gradient appears random.
Finally, the aggregated gradient produced by the model updating component is equal to the sum of non-private gradients plus the required DP noise (i.e., $ \sum_{i=1}^{n} g_i + \xi $), providing the same privacy-utility trade-off for the model as DP-SGD~\cite{abadi2016deep}.
Therefore, our DP-masking mechanism is not only efficient for secure aggregation, but also defeats various privacy attacks against FL (e.g., data reconstruction,
membership inference).

Our DP-masking mechanism is part of \sys's trusted service code (\S\ref{sec:design:overview}).  
Our sandboxing mechanism will prevent DP-masking from being bypassed even when the model and training code provided by the model owner cannot be inspected and may be malicious (\S\ref{sec:sandbox}), and the integrity mechanism protects the correct execution of DP-masking (\S\ref{sec:integrity}).

\section{VMPL-Enforced Sandboxing}
\label{sec:sandbox}

ML models and associated training code are intellectual properties of the model owners.
Protecting the confidentiality of these model assets, however, prevents anyone from inspecting them.
As a result, \sys must regulate their data flow at runtime.
\sys sandboxes the untrusted/uninspectable model and data handling code, ensuring that each sandbox gets only a single data item (\S\ref{sec:privacy_barrier:item}) without any shared state or communication between sandboxes and that each sandbox's output passes only through our 
privacy barrier.
Note that, we do not sandbox the model updating code as it receives and operates on only the DP-masked single-item gradients.

\subsection{Overview}
\label{sec:sandbox:overview}

\sys's sandboxing mechanism builds upon Virtual Machine Privilege Levels (VMPLs) provided by AMD SEV-SNP~\cite{amd-sev,vmpl}.
Specifically, the sandbox must control both the input to the potentially malicious data handling code (provided by model owner) and its output.
The input consists of the current iteration's model and exactly one data item (as required in \S\ref{sec:privacy_barrier:item}).
The output is a gradient computed from the data item, which is then passed to our 
privacy barrier.
We note that there have been systems trying to achieve sandboxing via various TEE techniques~\cite{ryoan, park2025pave, mei2024cabin, zhang2025erebor}. These mechanisms are either restrictive or inefficient for isolating untrusted model and training code from the rest of our system.

In \sys, the trusted \sys service runs at the highest privilege level VMPL0 
(Figure~\ref{fig:vmpl-sandbox}).
The service code runs at ring~3 (user space), and the sandbox monitor 
runs at ring~0 (kernel space).
Here, the sandbox monitor controls hardware page permissions and enforces restrictions.
In contrast, the untrusted data handling code runs at a lower privilege level VMPL1 at ring~3, along with a minimal trusted proxy kernel running at ring~0, handling syscalls and exceptions as well as mediating the communication with the sandbox monitor.
Isolation between concurrent VMPL1 processes relies on the trusted proxy kernel and protected process mappings; the sandbox monitor at VMPL0 enforces the Reverse Map Table (RMP) permissions and request policies.

\sys uses only two VMPLs (out of four available~\cite{amd-sev,vmpl}) to separate trusted
enforcement from untrusted execution. 
Using additional
privilege levels introduces more VMPL transitions 
and mediation steps 
with performance penalties.
Our design avoids these overheads and instead concentrates policy enforcement in the sandbox monitor at VMPL0 by assigning request validation and memory permission control to it.
In addition, using a two-level privilege separation opens the way to apply 
similar sandboxing mechanisms~\cite{zhang2025erebor} with other TEE designs (e.g., Intel TDX \cite{intel-tdx}, ARM CCA \cite{arm-cca}).

\begin{figure}[t]
	\centering
    \includegraphics[clip=true, trim={0, 0, 0, 0}, width=\linewidth]{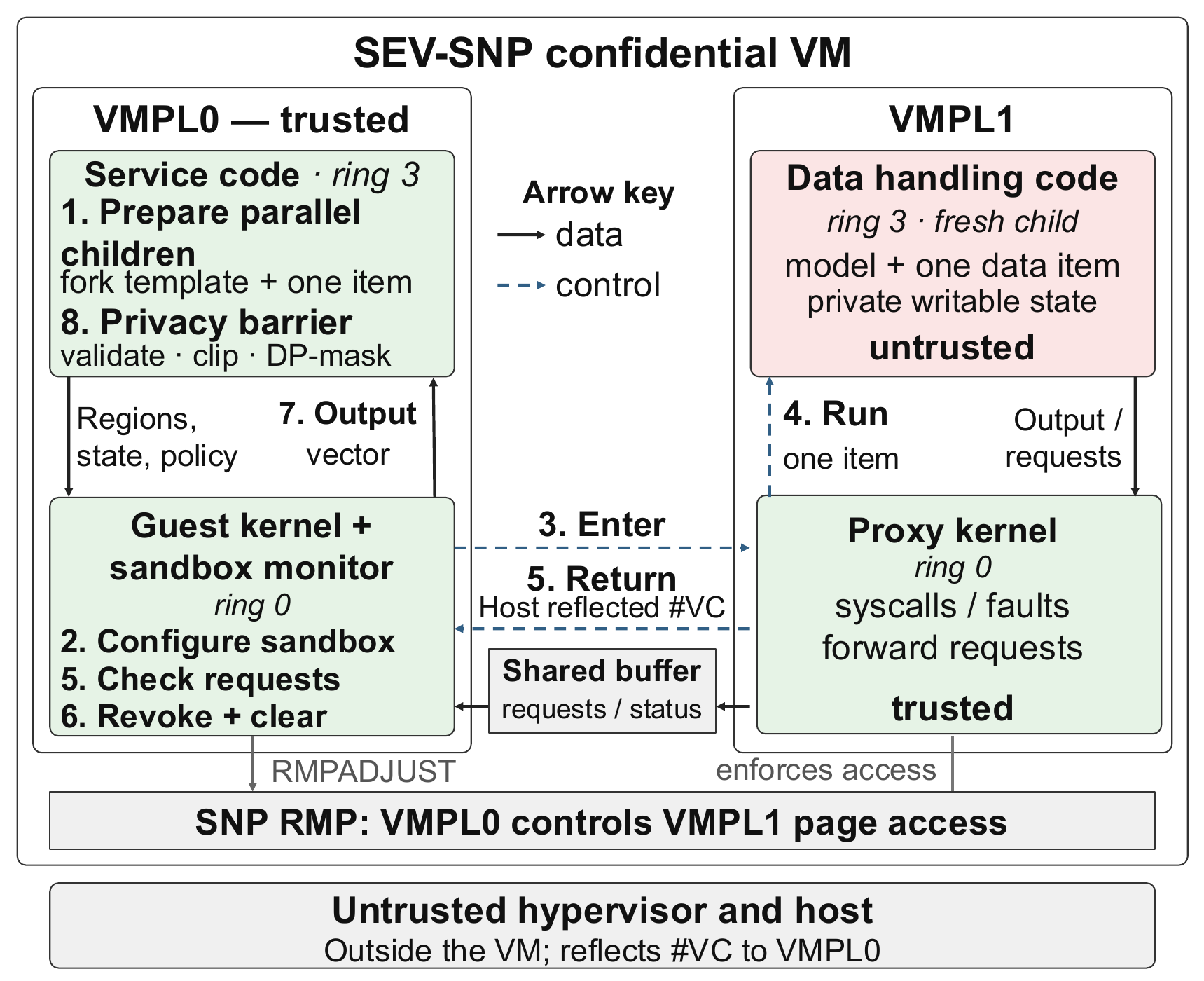}
	\caption{The architecture of VMPL-enforced sandboxing.
    }
	\label{fig:vmpl-sandbox}
\end{figure}

\subsection{Operational Workflow}
\label{sec:sandbox:operation}
The trusted service code and sandbox monitor have distinct responsibilities at
VMPL0. 
The service code, running in the user space, supplies the model and data items, and applies the privacy barrier to the output.
For each data item in a batch, the service code uses OS interfaces to create concurrent child processes to run the data handling code.
In particular, the service code first prepares the current iteration's
model and a reusable \emph{template} with trusted setup code that loads the data handling code, but without processing any data items. 
It forks concurrent child processes from the template, supplies one data item to each, and specifies a policy with the allowed actions (\textbf{Step 1} in Figure~\ref{fig:vmpl-sandbox}).
The concurrency is bound to the available compute resources, so that enough capacity is left for the sandbox monitor to handle requests and timeouts.

The service code then submits each child process' memory regions, entry state and policy 
to the sandbox monitor, running in the kernel space also at VMPL0.
This sandbox monitor controls VMPL transitions, page permissions and policy enforcement regarding the capabilities of the untrusted data handling code.
After receiving the freshly prepared child process to handle a single data item, the sandbox monitor records the page permissions, stores the policy and
initializes the Virtual Machine Save Area (VMSA) holding the
VMPL1 register state (\textbf{Step 2}).
Finally, it transitions to VMPL1, where the proxy kernel at VMPL1's ring~0 gets invoked (\textbf{Step 3}).

The proxy kernel transfers control to the child process  at ring~3, which computes over the single data item with the current model
to produce the gradient output (\textbf{Step 4}).
This output is returned via the proxy kernel to the sandbox monitor:
The proxy kernel places requests
in a shared communication buffer with the sandbox monitor. 
With SEV-SNP's \texttt{ReflectVC} enabled, executing a \texttt{CPUID} instruction
triggers a \#VC exception that the (untrusted) hypervisor just reflects to VMPL0, 
returning control to the sandbox monitor, where it checks whether the operation is permitted (\textbf{Step 5}).

Once the sandbox monitor determines the completion status, 
the monitor revokes the finished process' page permissions and clears its execution state (\textbf{Step 6}).
The output is then passed to the service code (\textbf{Step 7}), which validates the output and applies the privacy barrier (\textbf{Step 8}).
At the same time, the service code reaps the child process and the guest kernel reclaims its private address space.

\subsection{Isolation Enforcement and Output Control}
\label{sec:sandbox:isolation}

The trusted components at VMPL0 and VMPL1 together enforce the required isolation and output control.

\noindent\textbf{Policy.}
The policy for executing untrusted data handling code is determined by the service code. By default, it denies socket and file access as well as inter-process communication, except for one to interface with the sandbox monitor.
For permitted operations, the argument checks
restrict writes to private diagnostic pipes and prevent memory operations from
overlapping with the communication buffer.
The sandbox monitor checks each request and its relevant arguments
against the private policy copy fixed before the process execution. 
The proxy kernel’s copy provides an additional layer of defense and serves as an execution aid for performance.

The sandbox monitor also distinguishes operations which can be forwarded from
those that must be emulated. For instance, 
a signal handler would modify VMPL0 executor's state,
so this (untrusted) handler is answered without being installed.
The sandbox monitor also sets up a timer for each process with a watchdog to interrupt timed-out ones.

\noindent\textbf{Memory Isolation.}
Before entering VMPL1, the sandbox monitor grants access to the memory regions
specified for a process invocation, including the current model, the data item, and
the ML system/runtime code. It translates these regions to physical frames, and uses
\texttt{RMPADJUST} to set their VMPL1 permissions in SEV-SNP's RMP. These permissions constrain access independent of the
sandbox's page tables, keeping VMPL0 memory outside the granted regions.
The sandbox monitor records each granted frame, so teardown can release its
permissions and pinned-page references. The read-only template pages may remain
shared, while writable state must be private to each invocation. The shared
page accounting described later will prevent teardown from revoking pages
still in use by another sandbox.

\noindent\textbf{Output Control.}
On return, the trusted service code checks the gradient output against the expected tensor
shapes, numeric types, finite-value requirement, and total-size limit. These
checks validate its representation, regardless of how the untrusted data handling code
computed it. 
If any anomaly occurs (e.g., an invalid output or a time-out), the service code substitutes the output with a zero vector.  This does not introduce privacy issues because, afterwards, the service code will apply the privacy barrier, while
the completion status and diagnostic messages stay private in the service code. 
Together with the aforementioned policy, this makes the privacy barrier the \emph{only} permitted path for releasing
the computed gradient within our threat model (\S\ref{sec:threat}).%

\subsection{Performance Optimizations}
\label{sec:sandbox:optimizations}

Creating a fresh sandbox for each single data item is necessary (\S\ref{sec:privacy_barrier:item}) but expensive.
Here, we describe our optimizations preserving the
same isolation and output restrictions.

\noindent\textbf{Template Reuse.}
A persistent template retains reusable runtime state across iterations.
Before each iteration, the service code waits for active sandboxes to end and
loads the updated model into the template. Thus, all invocations in that iteration start with the same model; each receives one data item in a
fresh sandbox and is forced to discard its private state on completion. Reusing the
template  prepared independent of data items avoids the repeated runtime
initialization.

\noindent\textbf{Copy-on-Write.}
Besides immutable template contents, sandboxes also obtain private copies
of writable pages either eagerly before execution or lazily on first write,
decided by the service code based on the process' resident memory footprint.
Lazy copying avoids copying untouched memory, while eager copying avoids repeated
fault handling and contention during execution.
Appendix~\ref{sec:sandbox-optimizations} describes more details.

\noindent\textbf{Shared Page Permission.}
Immutable pages remain shared across sandboxes. 
The sandbox monitor tracks their use with a per-frame reference count: 
Granting a shared frame increments the count, 
and teardown removes its VMPL1 permissions only when the count reaches
zero, avoiding the invalidation of a page that is still in use by a sandbox.
The sandbox monitor also partitions permission-management locks by physical frame,
so unrelated updates do not serialize on a global lock. 
Writable state remains private, and teardown revokes access that is no longer needed.

\noindent\textbf{GPU Isolation with MPS.}
We utilize confidential GPUs for further optimization.  Specifically,
we use NVIDIA's Multi-Process Service (MPS)~\cite{nvmpsarch} on a
confidential GPU
to execute kernels from concurrent sandboxes.
MPS clients retain separate GPU-MMU-enforced isolated address spaces, with private
weights, activations, gradients, and temporary buffers.
In \sys, a fresh MPS client processes a single data item and returns its output.
The trusted service code validates its output and applies the privacy barrier.
The GPU driver
and MPS runtime are trusted inside the confidential VM.

\section{\sys Service Integrity}
\label{sec:integrity}

Unfortunately, TEE protections alone do \emph{not} establish that the intended service
environment is deployed~\cite{wilke2024snpguard}.
\sys' privacy barrier and sandboxing mechanisms depend on not only the integrity of the \sys service
code which implements these mechanisms, but also its runtime environment.
Our design builds on the attestation and image-integrity schemes in Confidential Containers (CoCo)~\cite{coco,coco-architecture},
which can use a Confidential VM (CVM) to house a deployed container.
The container image, however, may be stored in the untrusted cloud infrastructure.
As a result, we need to ensure the integrity of not only the CVM hosting \sys  components (as service containers) but also their container image integrity, before deployment and during runtime.

To do so, we combine the attestation of the CVM's guest kernel, sandbox monitor and proxy kernel, 
with a runtime policy that is enforced inside the CVM by CoCo's policy agent whose hash is bound to the attestation evidence.
This policy specifies the approved root hashes of the \sys service container image layers and the operations that are permitted,
restricting management access (e.g., \texttt{ExecProcess-} \texttt{Request}) to prevent altering of the service environment at runtime.
Using \texttt{tardev-snapshotter}~\cite{msftcctardev,msftkata}, we convert image layers into block devices with hash trees;
the policy agent checks their root hashes before mounting, and \texttt{dm-verity} verifies block devices as they are read in the CVM, detecting any tampering~\cite{dm-verity}.
In other words, while a remote attestation is limited to authenticate the CVM's initial environment only, our extra policy enforcement protects it during runtime.
As a result, participants can verify both the software and policies governing its execution~\cite{azure-aci-cc-attestation,azure-aks-cc-policy}.
This verification allows \sys's KDS to release keys only after verifying that the runtime environment matches the participants' approved configuration,
while also allowing them to outsource this task to KDS after its remote attestation (\S\ref{sec:design:workflow}).

\section{System Security Analysis}
\label{sec:security_analysis_system}

This section analyzes how \sys protects the confidentiality of training assets,
the privacy of individual data items, and the integrity of its service.
Appendix~\ref{sec:security_analysis_formal_dp}
provides the formal DP analysis of our privacy barrier mechanism.

\subsection{Confidentiality of Dataset and Model Assets}
\label{sec:analysis:confidentiality}
Dataset owners and model owners encrypt their respective assets
before uploading them to the untrusted cloud storage (Figure~\ref{fig:overview}).
These assets remain encrypted in storage and in transit 
to and among the training components.
They are decrypted only
inside authorized TEE components, where TEE memory protection prevents the
infrastructure and other participants from accessing them.
The asset owners provide the decryption keys to the key distribution service (KDS) in a TEE only after remotely attesting it.
The KDS releases the relevant keys
to training components only after verifying their service environments and
agreed session configurations (\S\ref{sec:design:workflow}, \S\ref{sec:integrity}).
Combined with sandboxing (\S\ref{sec:sandbox}),
these protections preserve the confidentiality
of dataset and model assets, without the inspection of the model or training
code.

\subsection{Privacy of Individual Data Items}
\label{sec:analysis:privacy}

The VMPL-enforced sandboxing 
restricts what the data handling code releases
after processing a private data item, forcing its output to pass through the privacy
barrier (\S\ref{sec:sandbox:isolation}).

\noindent\textbf{Attack Resistance.}
Compared to standard DP-SGD,
\sys's privacy barrier limits the information learned about an individual data item
from the released gradient and trained model.
Single-item processing and sandboxing together protect against the amplification of a sensitive item's influence, by deliberately restricting the data handling code to one item at a time (\S\ref{sec:privacy_barrier:item}) and preventing the shared state and communication among sandboxes (\S\ref{sec:sandbox:isolation}).
This  limits each item to contribute at most once per iteration.
The status (i.e., success or failure) is retained in the trusted service code in TEE.

In addition, dynamic gradient clipping (\S\ref{sec:privacy_barrier:clipping}) further limits the %
contribution of a data item.
The trusted service code in the data handling component computes gradient norms and histograms from single-item
outputs and keeps these intermediate values inside TEEs, while the clipping bound is
selected from a noisy aggregate histogram (\S\ref{sec:privacy_barrier:clipping}).
Privacy accounting includes both this selection and the subsequent gradient
releases (formalized in Appendix~\ref{sec:security_analysis_formal_dp}).
Together with DP-masking (\S\ref{sec:privacy_barrier:dpmasking}), these mechanisms protect the privacy of individual data items against gradient manipulation, data reconstruction and membership inference attacks.

\noindent\textbf{Collusion.}
The aforementioned protections hold even in the presence of collusion between the model owner and some dataset owners:
Without collusion, the potentially malicious model updating code receives random-looking gradients via DP-masking. 
With collusion of up to $n-1$ dataset owners and the model owner, the unknown gradients are still protected by the clipped gradients and DP, as the DP-masks are generated in the trusted admin component and applied in the data handling components enforced by the trusted service code and sandboxes within them (\S\ref{sec:privacy_barrier:dpmasking}, \S\ref{sec:sandbox:isolation}).
Other collusions with the \sys provider or cloud provider cannot bypass these protections under the TEE and integrity assumptions (\S\ref{sec:threat}).

\noindent\textbf{Training Session Manipulation.}
The trusted admin component
enforces the agreed privacy budget and stops training when it is exhausted.
The KDS checks the session configuration before releasing keys, and discards
these keys after provisioning the authorized training components
(\S\ref{sec:design:workflow}); therefore, a new training session with the same datasets cannot be run without the renewed approvals from dataset owners.
A rollback of a previously authorized session by the KDS is outside our threat model (\S\ref{sec:threat}).

\subsection{Service Integrity and Sandbox Enforcement}
\label{sec:analysis:integrity}
The trusted service code, sandbox monitor and guest kernel run inside TEE at the highest privilege level (VMPL0),
whereas the potentially malicious data handling code is launched with a proxy kernel at a lower privilege level (VMPL1). We do not need to sandbox the model updating code (\S\ref{sec:sandbox}), as it receives and operates on only the DP-masked single-item gradients.
The trusted proxy kernel preserves each sandbox's page-table isolation to prevent unauthorized memory access between concurrent VMPL1 sandboxes; VMPL separation further protects VMPL0 from unauthorized accesses from VMPL1.
The sandbox monitor's policy is the final authority on allowed operations for the data handling code (\S\ref{sec:sandbox:isolation}),
enabling the service code, sandbox monitor and guest kernel to remain trusted at runtime.

The attestation and image-integrity mechanisms allow the KDS 
to verify 
(on behalf of dataset owners and model owner after remote attestation) 
that \sys service components are not altered before deployment.
The policy pinning the allowed measurements of these components is part of the attestation evidence used by the KDS (\S\ref{sec:integrity}).
These measurements can be computed by any actor as \sys components and the KDS are available for inspection (\S\ref{sec:threat}, \S\ref{sec:design:overview}).
Attacks on TEEs (e.g., covert- and side-channel, physical, rollback attacks) remain outside the threat model (\S\ref{sec:threat}).

\section{Implementation and Evaluation}
\label{sec:evaluation}

We implemented all \sys training components (i.e., admin, data handling, model updating) 
as Python services via Confidential Containers (CoCo)~\cite{coco},
as well as the management and key distribution services with respective Python SDKs.
Besides the Python counterparts in the data handling component, 
the VMPL-enforced sandboxing comprises a Linux kernel module (i.e., sandbox monitor) and a proxy kernel for the data handling code~\cite{mei2024cabin}.
We implemented all our optimizations for sandboxing (details in Appendix~\ref{sec:sandbox-implementation}).
The integrity mechanism builds on CoCo's attestation, runtime-policy, and
image-verification features
(\S\ref{sec:integrity}). 

We run our implementation in QEMU with patched SEV-SNP host and guest kernels.
We evaluate \sys 
along four dimensions: 
\textbf{1)} the impact of the privacy barrier on model accuracy and training time (\S\ref{sec:evaluation-dp}),
\textbf{2)} the overhead of VMPL-enforced sandboxing for single-item processing and the benefits of our optimizations (\S\ref{sec:evaluation-sys}),
\textbf{3)} the scalability with different batch sizes, numbers of dataset owners, and confidential GPUs (\S\ref{sec:evaluation:scalability}), and
\textbf{4)} the comparison with the state-of-the-art solutions that cannot offer the same simultaneous protections, in contrast to \sys (\S\ref{sec:evaluation-comparison}).

In our evaluation, we use two models.  The first one is
a 6-layer CNN model with CIFAR10 dataset (denoted as CIFAR10-CNN6), which is often used or similar to the models used to evaluate many privacy-preserving training systems~\cite{mo2021ppfl, mo2020darknetz, liu2024pencil, patra2020blaze}.
To show \sys' wider applicability, we also use a pre-trained 12-transformer-layer model `Roberta-base' with AG-News dataset (denoted as AGNEWS-Roberta), and then fine-tune it.
Although this second model is still smaller than today's large language models, it matches the scale of many realistic models used in various industrial domains~\cite{Uni-Mol2,Chronos-2,TimesFM,Kronos-model}; it is also already \emph{larger than or at the same level} as some of the largest models evaluated in the context of privacy-preserving collaborative training (e.g., Wide-ResNet-40 and NFNet-F3~\cite{de2022unlocking}).
We train the two  models with 1000 iterations and a learning rate of 0.01.

Unless otherwise noted, for each training session, we launch one admin component, one model updating component, and 4 data handling components (corresponding to 4 dataset owners).
Each data handling component computes gradients for 32/64/128 data items in their \emph{local} batch,
resulting in a \emph{global} batch size of 128/256/512 data items per iteration.
We vary these settings when evaluating the system scalability.
Details of our setup are in Appendix~\ref{sec:evaluation-setup}.

\subsection{Impact of Privacy Barrier}
\label{sec:evaluation-dp}

Here, we first evaluate the impact of privacy barrier 
on model convergence and accuracy. We then confirm that the dynamic gradient clipping does not adversely affect convergence.
The overhead of creating, distributing and applying DP-masks with dynamic gradient clipping is also negligible (only about {$1\%$},  `+Opt' vs. `+PB' in Figure \ref{fig:sys-eval-overhead-breakdown}).

\noindent\textbf{Model Convergence \& Accuracy.}
The combination of our single-item processing (\S\ref{sec:privacy_barrier:item}) and DP-masking (\S\ref{sec:privacy_barrier:dpmasking}) mechanisms achieves stronger privacy guarantees than DP-SGD (\S\ref{sec:analysis:privacy}).
We show the trade-off between convergence/accuracy and privacy.
We train models
with different privacy budgets to compare how they converge at different privacy levels. 
The added noise is adapted, so that the full privacy budget is spent exactly after 1000 iterations.
As expected, 
Figures~\ref{fig:dp-convergence-iter2} and \ref{fig:dp-convergence-iter3} show that 
models reach higher accuracy with larger $\epsilon$ values (i.e., lower noise),
and also approach the accuracy of baselines with no privacy protection.

\noindent\textbf{Dynamic Gradient Clipping.} 
Figure~\ref{fig:dp-mnist-dynclip} shows the DP noise norm with dynamic gradient clipping (\S\ref{sec:privacy_barrier:clipping}), compared to the noise added with an optimized fixed clipping bound $C=0.5$.
Our dynamic gradient clipping approach gradually reduces the noise values,
leading to better model convergence and accuracy 
without privacy leakage.
In contrast to prior work~\cite{andrew2021differentially, golatkar2022mixed}, our dynamic approach avoids the lengthy tuning of the clipping bound before training the model, thus leading to high practicality.

\begin{figure}[tp]
    \centering
    \subfloat[CIFAR10-CNN6 convergence\label{fig:dp-convergence-iter2}]{
        \includegraphics[clip=true, trim=5 5 13 5, width=0.5\columnwidth]{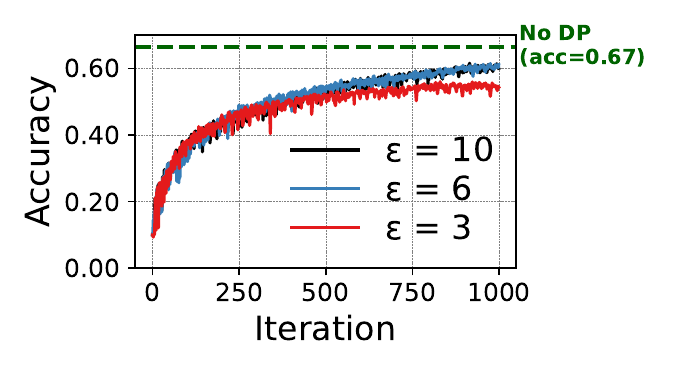}
    }
    \subfloat[AGNEWS-Roberta convergence\label{fig:dp-convergence-iter3}]{
        \includegraphics[clip=true, trim=5 5 13 5, width=0.5\columnwidth]{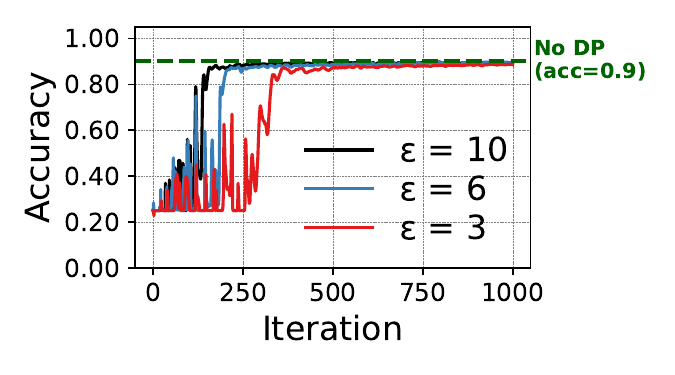}
    }
    \caption{Privacy barrier effects with various privacy levels.}
    \label{fig:dp-utility-convergence}
\end{figure}

\begin{figure}[tp]
    \centering
     {
    \includegraphics[clip=true, trim=5 5 13 5, width=0.67\columnwidth]{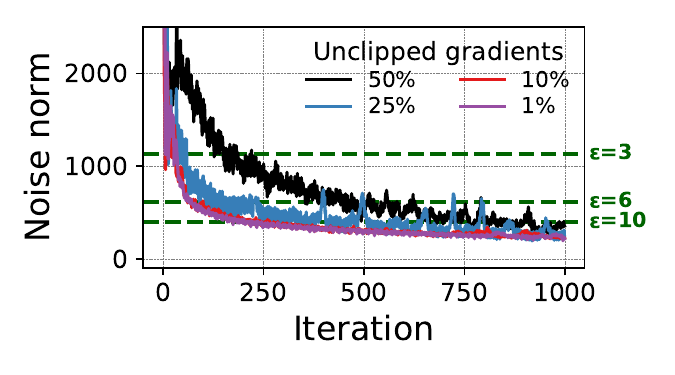}
     }
    \caption{Effects of dynamic gradient clipping.}
    \label{fig:dp-mnist-dynclip}
\end{figure}

\begin{figure}[t]
    \centering
    \subfloat{
        \centering
        \includegraphics[clip=true, trim={10, 74, 10, 10}, width=\linewidth]{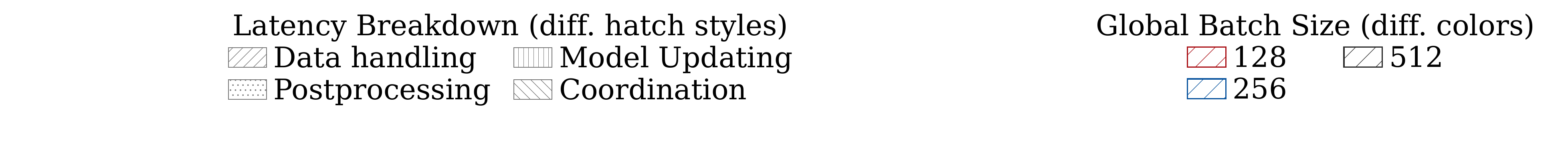}%
    }
    \setcounter{subfigure}{0}
    
     \subfloat[CIFAR10-CNN6]{
        \includegraphics[width=0.47\columnwidth]{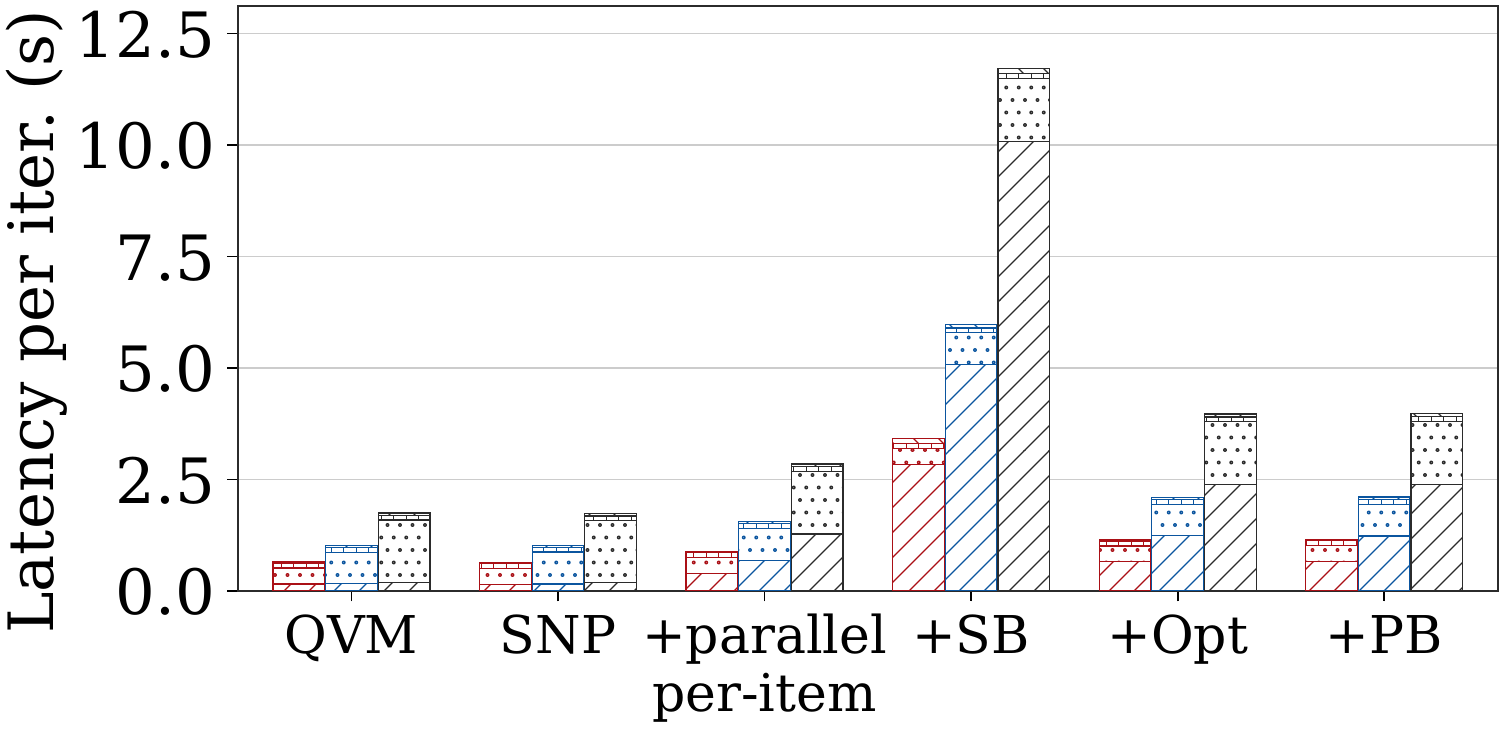}
        \label{fig:sys-eval-overhead-breakdown-cnn6}
    }
    \subfloat[AGNEWS-Roberta]{
        \includegraphics[width=0.47\columnwidth]{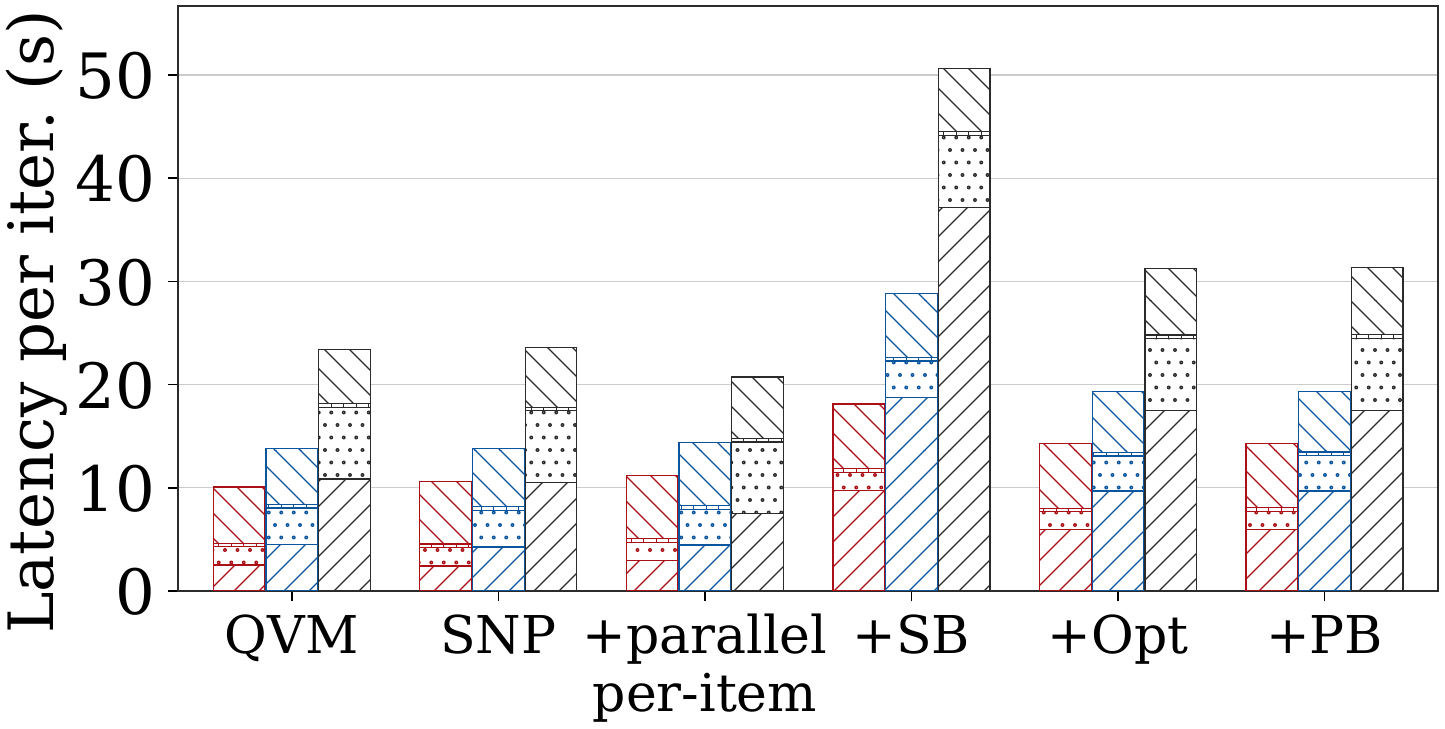}
        \label{fig:sys-eval-overhead-breakdown-agnews}
    }
    \caption{End-to-end latency breakdown. QVM is the normal QEMU VM, and SNP runs in an SEV-SNP CVM.
    }
   \label{fig:sys-eval-overhead-breakdown}
\end{figure}

\subsection{Overhead of VMPL-Enforced Sandboxing}
\label{sec:evaluation-sys}

We evaluate the performance overhead of our sandboxing mechanism to isolate untrusted data handling code.
Due to the single-item processing requirement for privacy, this overhead is critical for \sys' scalability.
We first show the overheads of SNP CVM, single-item processing and VMPL isolation, and then
present the benefits of our optimizations.

\noindent\textbf{Sandboxing Overhead.}
We demonstrate the sandboxing overhead by comparing different modes of processing data items with the data handling code.
Figure~\ref{fig:sys-eval-overhead-breakdown} shows that running in a CVM is not a big factor (`QVM' vs.\ `SNP'):
the difference is at most 5.0\% across both models and the three global batch sizes.
As expected, processing items one at a time in parallel without VMPL isolation (`+parallel, per-item') adds 38.7--64.5\% higher latency over the `SNP' batched processing for CIFAR10-CNN6.
For AGNEWS-Roberta, the difference ranges from 12.2\% lower to 5.9\% higher latency. 
At a global batch size of 512, 
single-item processing in parallel is faster because it avoids our privacy library (Opacus~\cite{opacus2021})'s per-item-gradient hooks, 
whereas the batch processing in `SNP' uses one process with BLAS threading.
This reduces the total data handling  time; thus, also the end-to-end latency.
The biggest overhead is caused by the VMPL isolation of the single-item processing without optimizations (`+SB'):
each item incurs a fresh process with model loading and up-front page population and privatization before VMPL switch.
At a global batch size of 512, latency rises from 2.85 to 11.71\,s for CIFAR10-CNN6 and from 20.73 to 50.61\,s for AGNEWS-Roberta when adding this isolation to per-item processing.

Our optimizations reduce this overhead.
Across global batch sizes, the optimized sandboxing (`+Opt') is
2.85--2.98$\times$ faster than VMPL isolation (`+SB') for CIFAR10-CNN6 and
1.27--1.62$\times$ faster for AGNEWS-Roberta
(Figure~\ref{fig:sys-eval-overhead-breakdown}).
Therefore, the final overhead across all three global batch sizes
over vanilla unprotected baselines with batching (i.e., `QVM', `SNP') is $74.9$--$129.7\%$ and
$32.6$--$42.1\%$ for CIFAR10-CNN6 and AGNEWS-Roberta, respectively.
We think this overhead is acceptable for the strong confidentiality and privacy protections \sys provides with its privacy barrier and sandboxing mechanisms, in a mutually distrustful environment.

\begin{figure}[t]
    \centering
    \subfloat{
        \centering
        \includegraphics[clip=true, trim={0, 20, 0, 20}, width=0.5\linewidth]{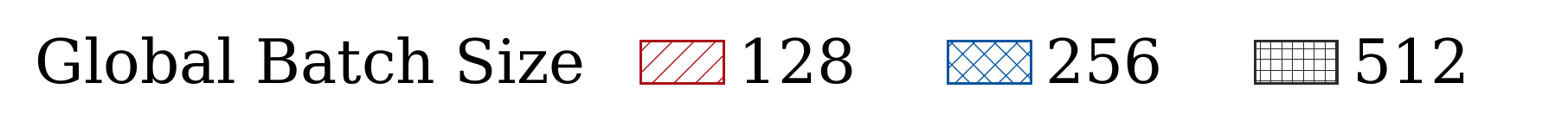}%
    }
    \vspace{-0.7mm}
    \setcounter{subfigure}{0}
    
     \subfloat[CIFAR10-CNN6]{
        \includegraphics[width=0.47\columnwidth]{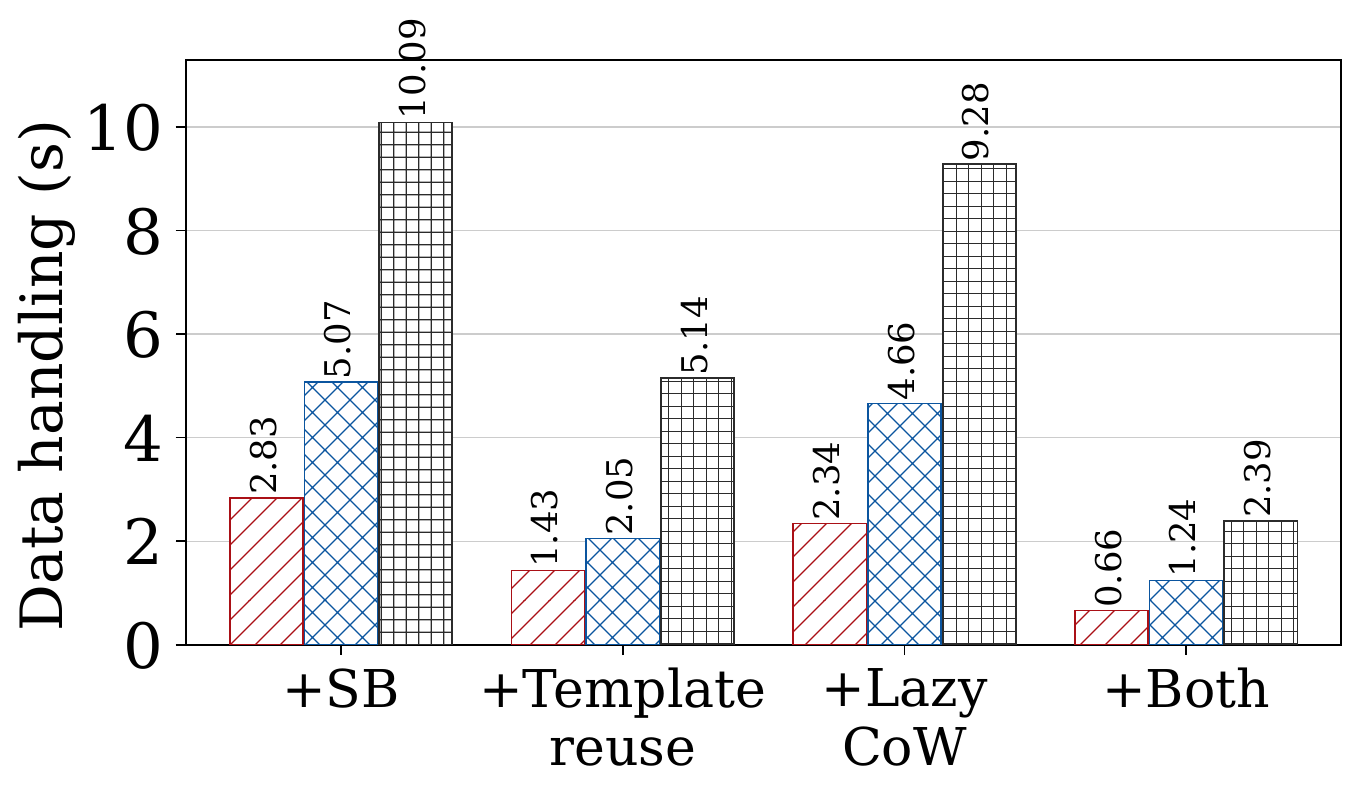}
        \label{fig:sys-eval-sandbox-ablation-cnn6}
    }
    \subfloat[AGNEWS-Roberta]{
        \includegraphics[width=0.47\columnwidth]{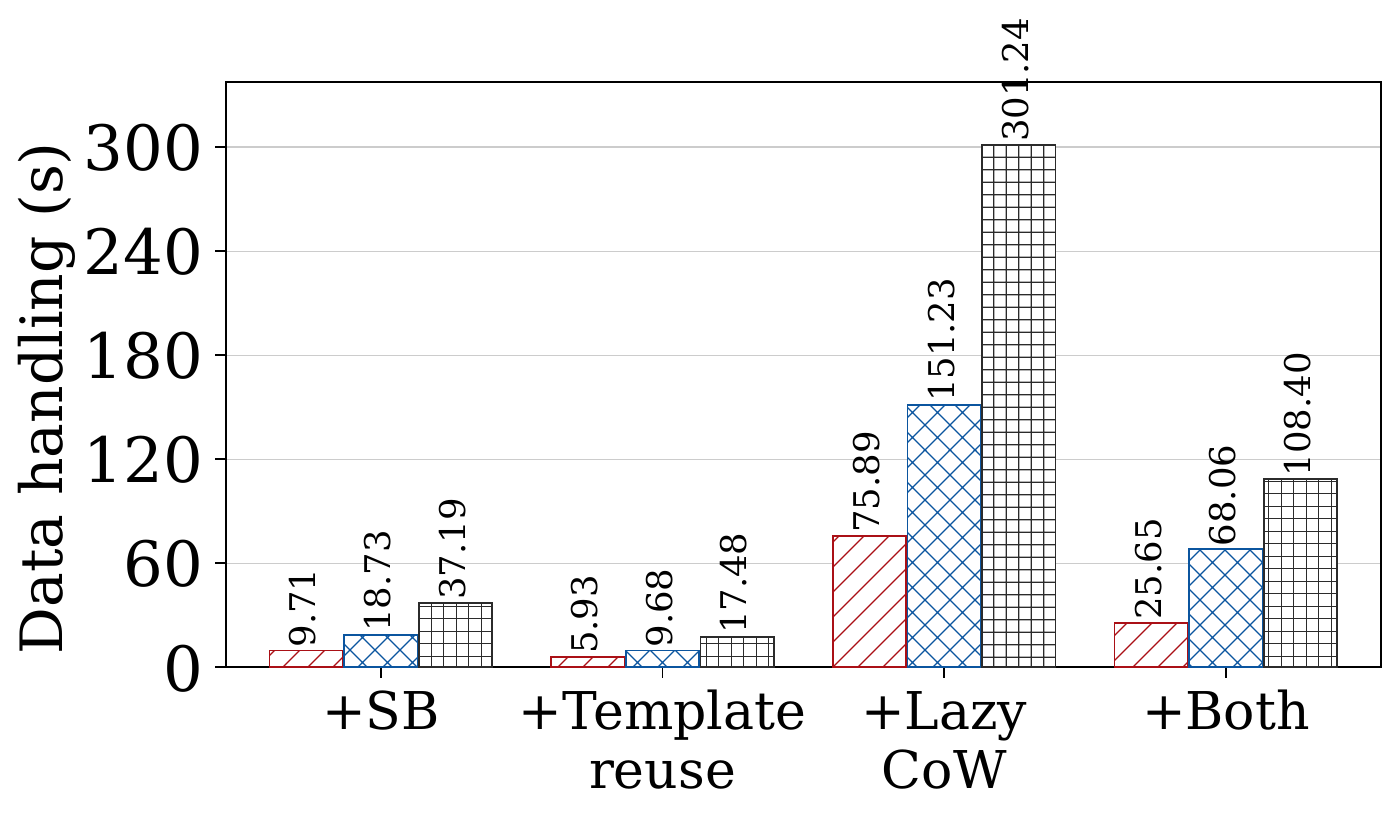}
        \label{fig:sys-eval-sandbox-ablation-agnews}
    }
    \caption{Ablation study of sandboxing optimizations.
    }
   \label{fig:sys-eval-sandbox-ablation}
\end{figure}

\noindent\textbf{Benefits of Optimizations.}
To understand the template reuse and copy-on-write,
we isolate their effects in
Figure~\ref{fig:sys-eval-sandbox-ablation}.  The shared-page permission management is a part of the default sandbox.
For CIFAR10-CNN6, persistent templates (`+Template reuse') alone accelerate  sandboxing (`+SB') by
1.96--2.48$\times$, while lazy copy-on-write (`Lazy CoW') alone yields
1.09--1.21$\times$ speedups. 
Combining both optimizations reduces the data handling time from 10.08 to 2.39\,s for a global batch size of 512, providing a speedup of
4.09--4.31$\times$. 

In contrast, a simple combination of optimizations does not suit AGNEWS-Roberta:
the model's large working set makes copy-on-write costly,
because write faults require VMPL1 exits as well as the page copying and permission updates by the sandbox monitor at VMPL0, repeated across 32 concurrent per-item processes.
These actions increase the latency to 301.21\,s (`+Lazy CoW'),
while the combination with template reuse (`+Both') yields 103.38\,s, still much worse than VMPL sandboxing (`+SB') with 37.16\,s.
In this case, our optimization of reusing templates with private copies of writable pages created before VMPL1 entry (`+Template reuse') avoids repeated model loading and copy-on-write fault handling during VMPL1 execution, thus yielding the best performance by reducing the data handling time to 17.46\,s, a speedup of 1.64--2.13$\times$ for a global batch size of 512.

\begin{figure}[t]
    \centering
    \subfloat{
        \centering
        \includegraphics[clip=true, trim={10, 73, 10, 10}, width=\linewidth]{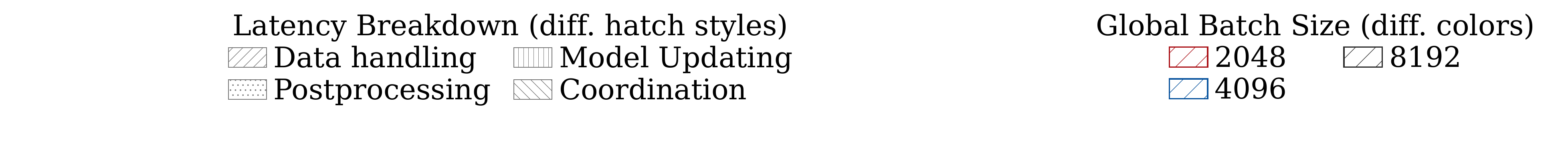}%
    }
    \setcounter{subfigure}{0}

     \subfloat[CIFAR10-CNN6]{
        \includegraphics[width=0.47\columnwidth]{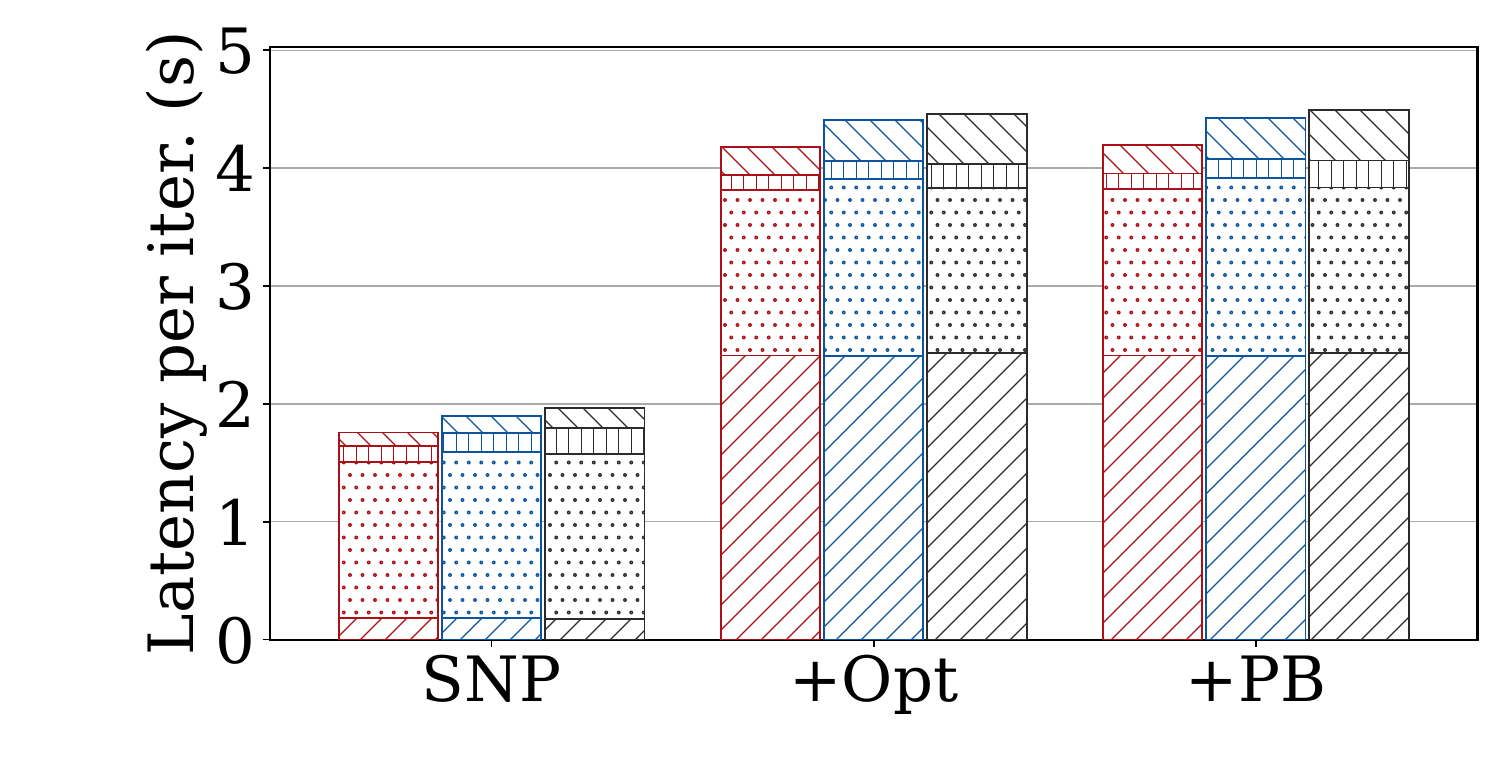}
        \label{fig:sys-eval-overhead-accumulate-cnn6}
    }
    \subfloat[AGNEWS-Roberta]{
        \includegraphics[width=0.47\columnwidth]{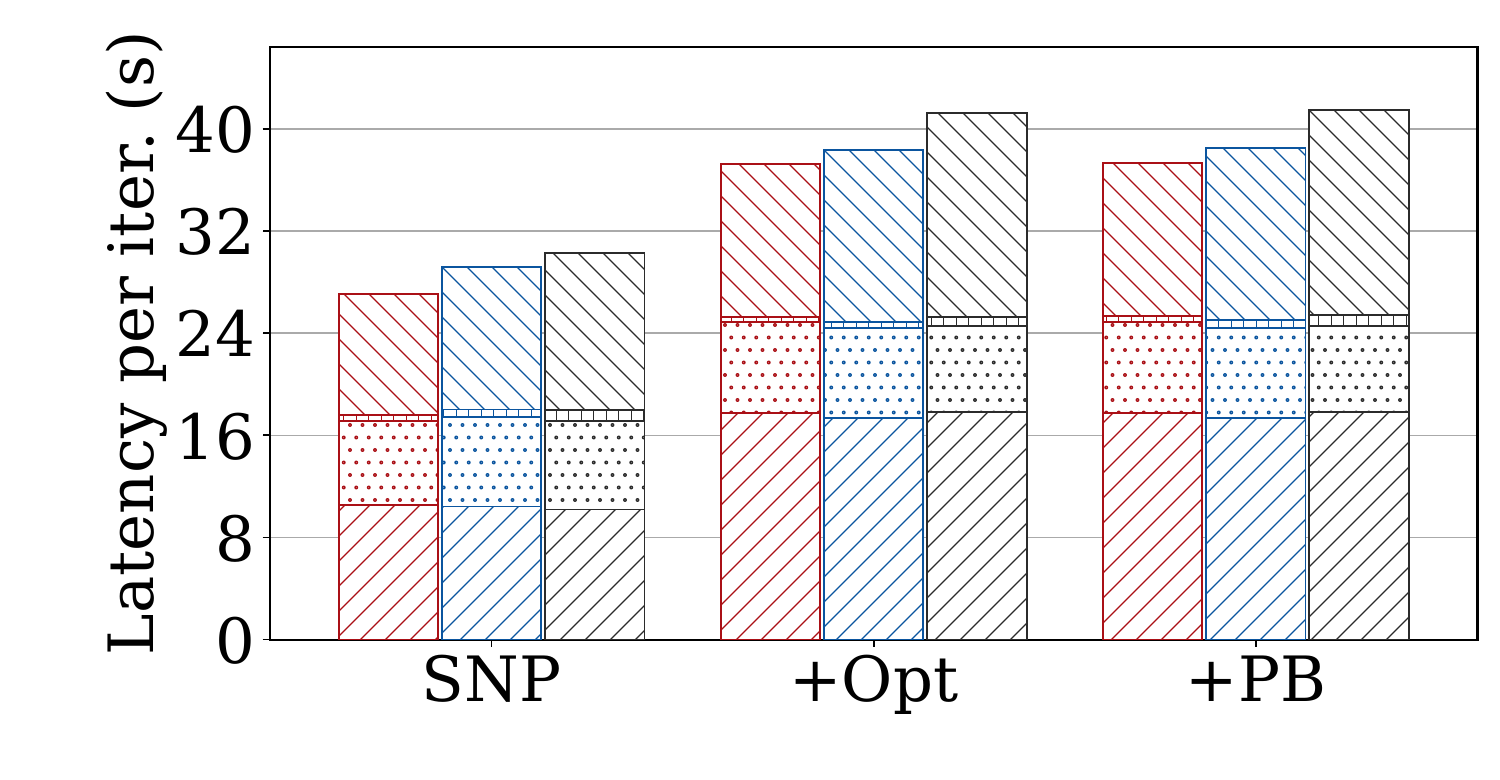}
        \label{fig:sys-eval-overhead-accumulate-agnews}
    }
    \caption{Latency with larger batch sizes with 4 dataset owners, each with 4/8/16 parallel data handling components. Local batch size is 128 for a global batch size of 2048/4096/8192.
    }
   \label{fig:sys-eval-overhead-accumulate}
\end{figure}

\subsection{Scalability}
\label{sec:evaluation:scalability}
We evaluate \sys' scalability with different batch sizes, numbers of dataset owners,
and the use of confidential GPUs.

\noindent\textbf{Large Batch Sizes and Datasets.}
DP training can benefit from larger batch sizes~\cite{mcmahan2017learning, de2022unlocking}.
In \sys, multiple parallel data handling components can be allocated to the processing of a large batch (\S\ref{sec:design:workflow}).
Figure~\ref{fig:sys-eval-overhead-accumulate} shows that the increased iteration latency remains modest, even with our sandboxing and privacy barrier mechanisms (i.e.,`+Opt' and `+PB'). 
The global batch size can be increased from 2048 to 8192 with an overhead of only $7.07\%$ (4.19 to 4.49\,s) for CIFAR10-CNN6 and $11.08\%$ (37.32 to 41.45\,s) for AGNEWS-Roberta, respectively.
This approach can also process large datasets by splitting them into disjoint data handling components (\S\ref{sec:design:workflow}). We omit the similar results due to space limit.

\begin{figure}[t]
    \centering
    \subfloat{
        \centering
        \includegraphics[clip=true, trim={10, 73, 10, 10}, width=\linewidth]{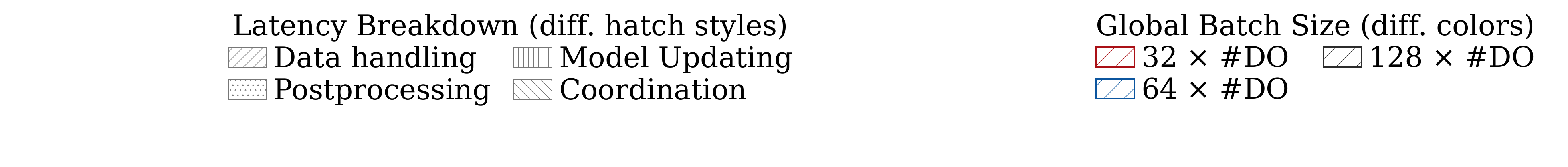}%
    }
    \setcounter{subfigure}{0}

     \subfloat[CIFAR10-CNN6]{
        \includegraphics[width=0.47\columnwidth]{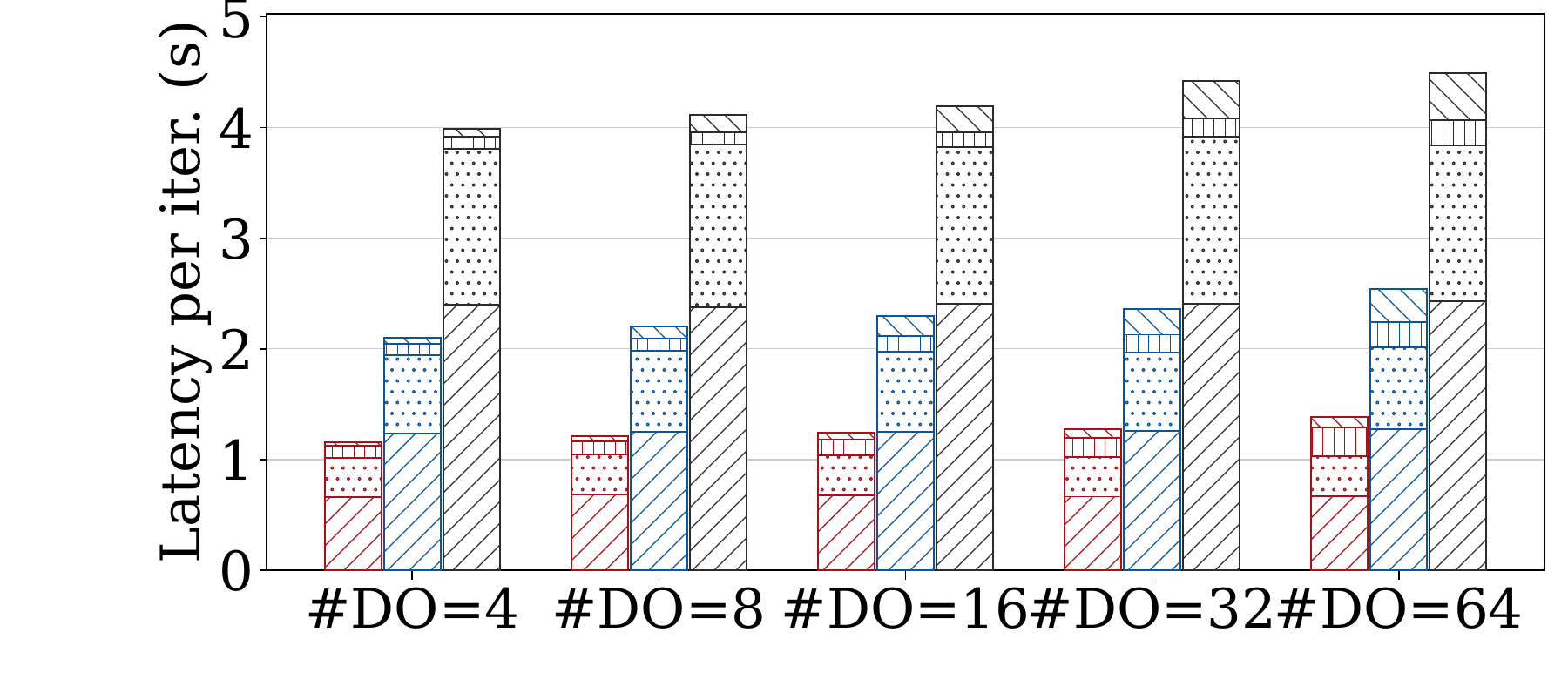}
        \label{fig:sys-eval-overhead-scale-cnn6}
    }
    \subfloat[AGNEWS-Roberta]{
        \includegraphics[width=0.47\columnwidth]{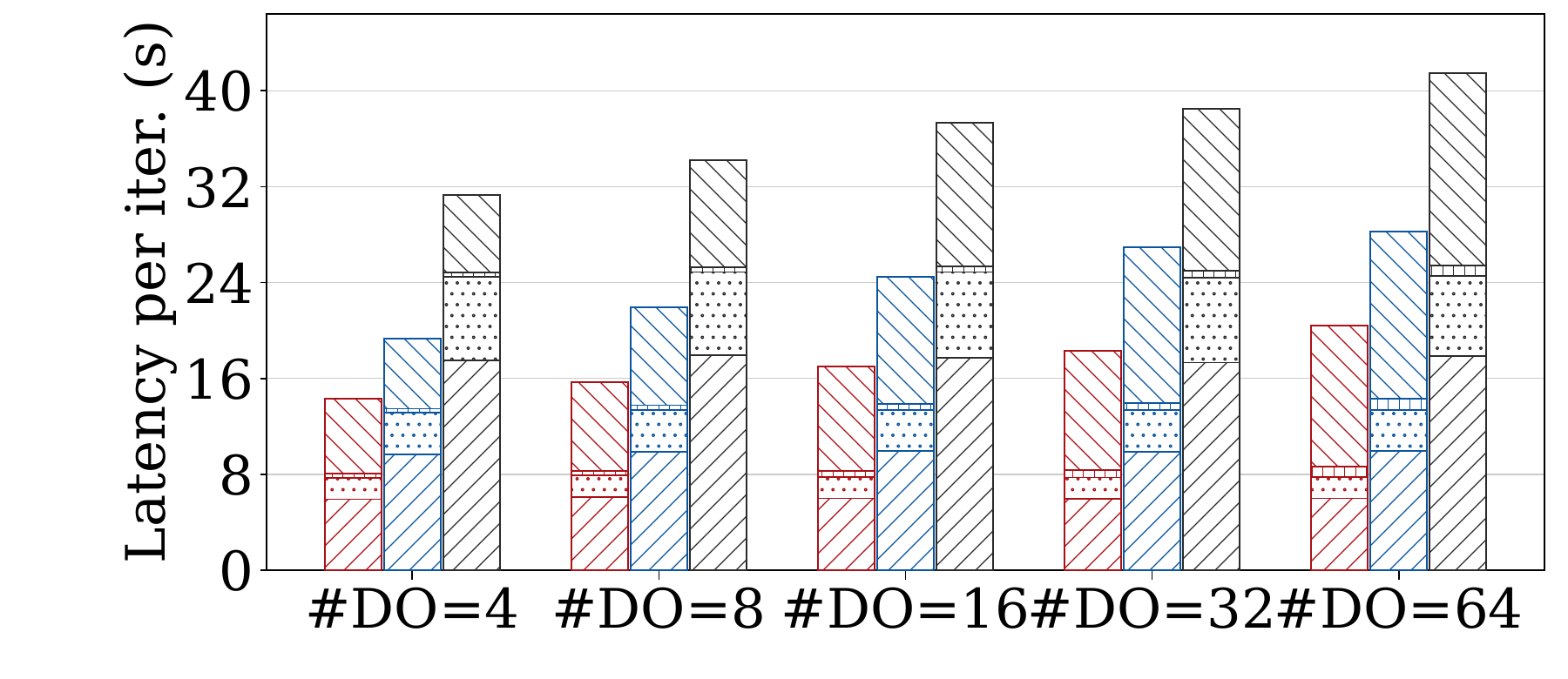}
        \label{fig:sys-eval-overhead-scale-agnews}
    }
    \caption{Latency with more dataset owners (`\#DO') with different local batch sizes of 32/64/128.
    }
   \label{fig:sys-eval-overhead-scale}
\end{figure}

\begin{figure}[t]
    \centering
    \subfloat{
        \centering
        \includegraphics[clip=true, trim={10, 73, 10, 10}, width=\linewidth]{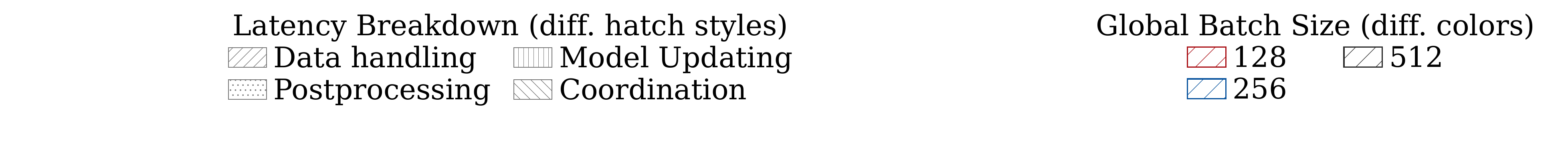}
    }
    \setcounter{subfigure}{0}
    \subfloat[CIFAR10-CNN6]{
        \includegraphics[width=0.47\columnwidth]{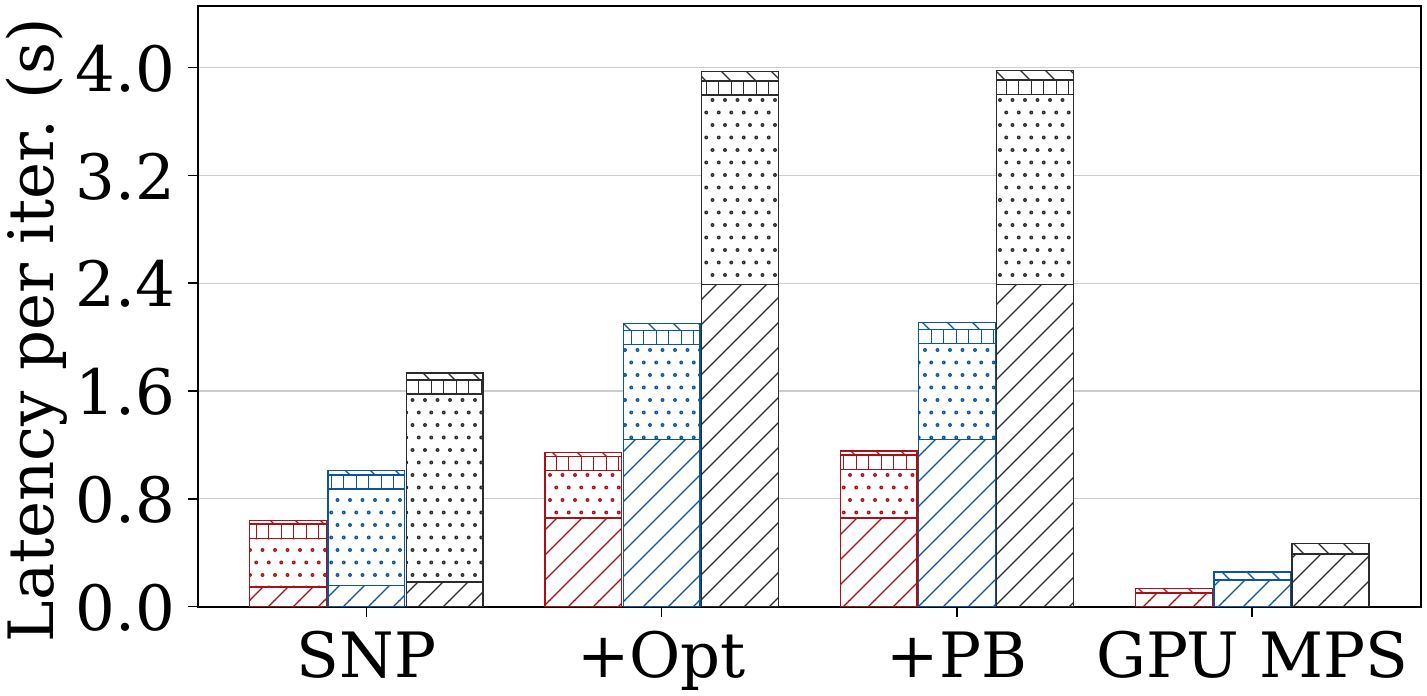}
        \label{fig:sys-eval-gpu-tee-cnn6}
    }
    \subfloat[AGNEWS-Roberta]{
        \includegraphics[width=0.47\columnwidth]{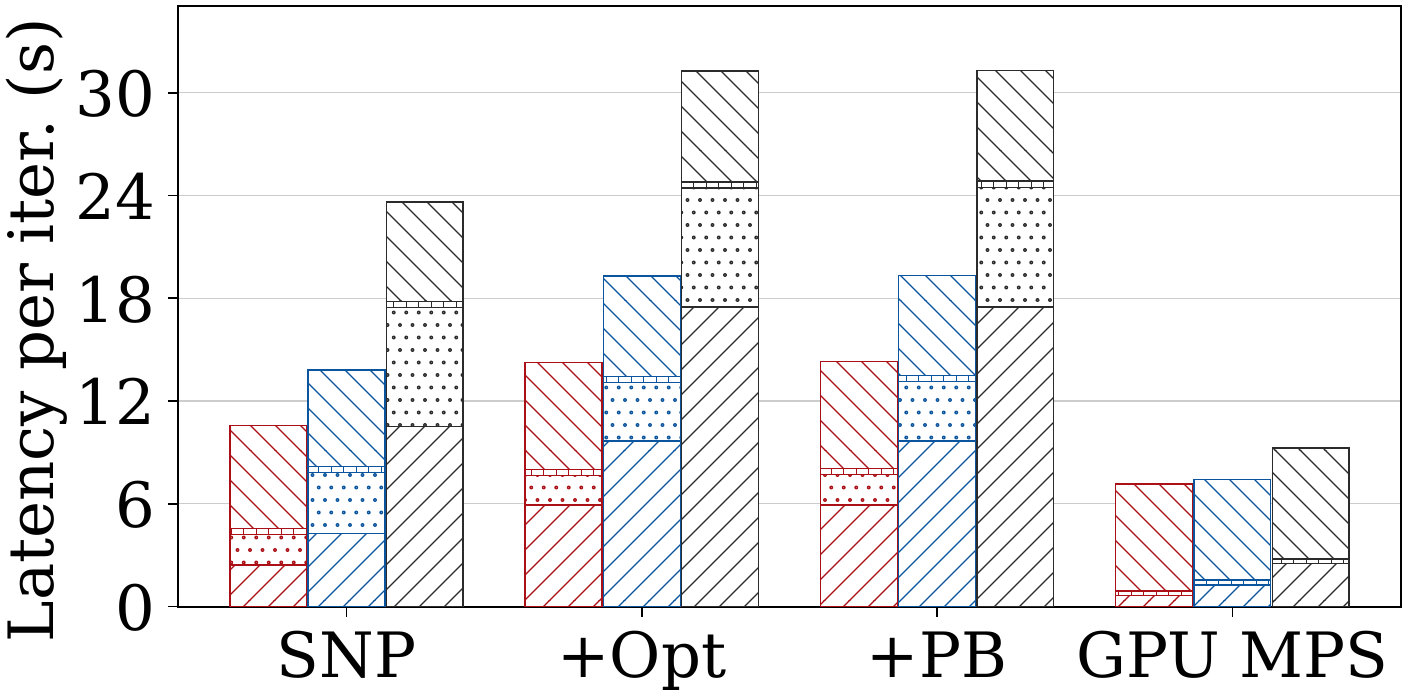}
        \label{fig:sys-eval-gpu-tee-agnews}
    }
    \caption{Latency with the confidential NVIDIA H100 performance. `SNP', `+Opt' and `+PB' are the CPU baselines. 
    }
    \label{fig:sys-eval-gpu-tee}
\end{figure}

\noindent \textbf{More Dataset Owners.} \sys targets scenarios where there are at most a few tens of organizational dataset owners (\S\ref{sec:actors}). 
Figure~\ref{fig:sys-eval-overhead-scale} shows the training latency increases sub-linearly with more dataset owners. For instance, scaling from 4 to 64 dataset owners (with a local batch size of 128) incurs a latency increase of $32.40\%$ (31.31 to 41.45\,s) for AGNEWS-Roberta, demonstrating \sys' efficient sublinear scaling.

\noindent\textbf{Confidential GPU Acceleration.}
Figure~\ref{fig:sys-eval-overhead-breakdown} shows that the data handling and model updating phases dominate the iteration latency, which could be accelerated by confidential GPUs.
We use 32-way MPS with one fresh private process
per data item, 
so the GPU hardware enforces their address space isolation but also executes them concurrently.
Figure~\ref{fig:sys-eval-gpu-tee}
shows that `GPU MPS' improves over `+PB' by $8.2$--$8.6\times$ for
CIFAR10-CNN6 and $2.0$--$3.4\times$ for AGNEWS-Roberta, representing a great potential of using confidential GPUs.

\subsection{Comparison with Other Systems}
\label{sec:evaluation-comparison}
We compare \sys with three representative privacy-preserving training systems:
1) commercial Azure Confidential Clean Rooms \textbf{Azure CR}\footnote{Azure CR does not support GPU TEEs as of this writing.}~\cite{azurecleanroom}, 
2) cryptography-based \textbf{Pencil}~\cite{liu2024pencil}, 
and 3) TEE-based \textbf{Citadel}\footnote{We adapt Citadel from Intel SGX to AMD SEV-SNP for a fair comparison.}~\cite{zhang2021citadel}.
These baseline systems offer different points in the design space, but
note that, in contrast to \sys,  they cannot simultaneously protect the confidentiality of datasets, model and training code as well as the individual data privacy.

Figure~\ref{fig:eval-sys-baseline} shows the training latency per iteration.
Compared to Azure CR (CPU), \sys (CPU) has higher latency for CIFAR10-CNN6 and slightly lower latency for the large AGNEWS-Roberta.
\sys (CPU) is moderately slower than Citadel (CPU), with overheads of $32.6$--$40.0\%$ for AGNEWS-Roberta and $81.2$--$129.7\%$ for CIFAR10-CNN6.
This overhead is due to \sys's single-item processing and VMPL-enforced sandboxing for stronger privacy protections that these baselines lack.
Compared to both, \sys{} (GPU) is 3.74$\times$ faster than Azure CR (CPU) for AGNEWS-Roberta, and up to 2.55-3.68$\times$ faster than Citadel (CPU) for a batch size of 512.
Pencil (CPU) and Pencil (GPU) suffer from heavy cryptographic operations and do not support AGNEWS-Roberta's transformer layers.
For CIFAR10-CNN6, \sys (CPU) achieves up to $199.4\times$ speedup over Pencil (CPU)
and $34.9\times$ over Pencil (GPU), 
while \sys (GPU) is $298.4\times$ faster than Pencil (GPU).

\begin{figure}[tp]
    \centering
    \subfloat[CIFAR10-CNN6]{
        \includegraphics[width=0.49\columnwidth]{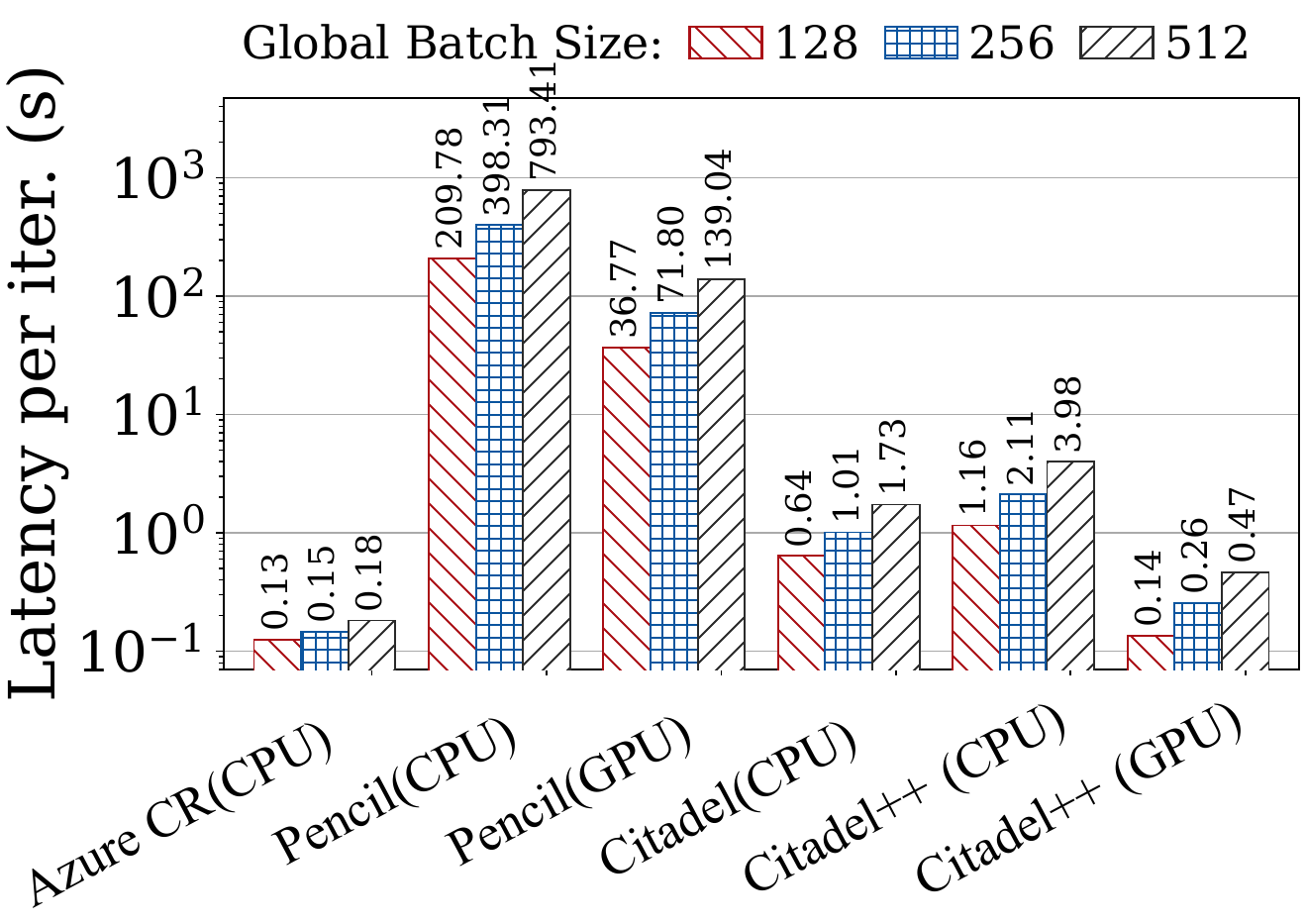}
        \label{fig:eval-sys-baseline-cnn6}

    }
    \subfloat[AGNEWS-Roberta]{
        \includegraphics[width=0.49\columnwidth]{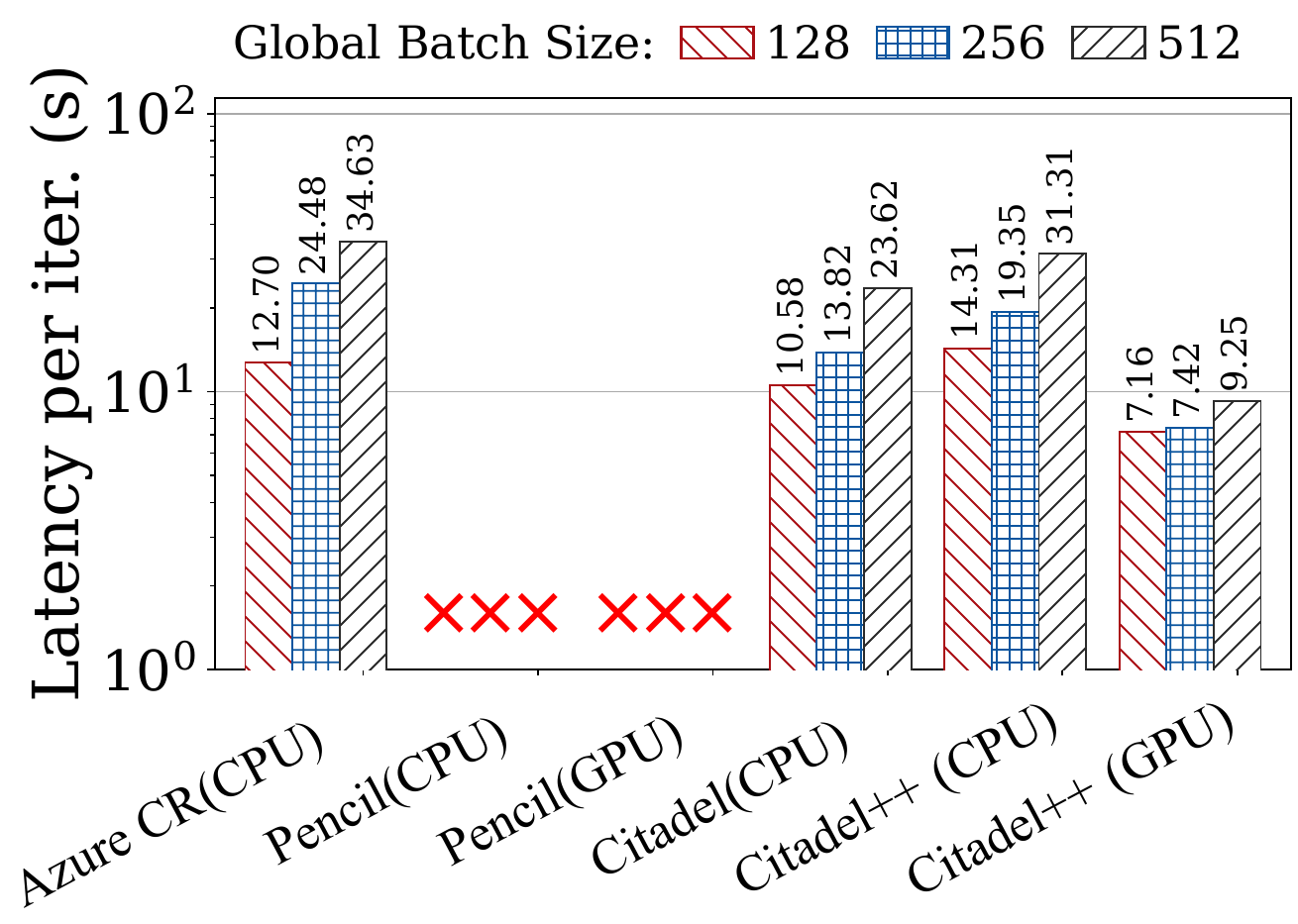}
        \label{fig:eval-sys-baseline-agnews}

    }
    \caption{
    Iteration latency for \sys and the state-of-the-art confidential training systems. Pencil does not support Roberta.
    `Pencil (GPU)' uses NVIDIA H100 in normal mode. `\sys (GPU)' uses the confidential-H100 MPS mode.
    }
    \label{fig:eval-sys-baseline}
\end{figure}

\section{Related Work}
\label{sec:related}

Confidential computing, especially its TEE realization, has been widely deployed in cloud infrastructure~\cite{apple-pcc, azure-cc,
google-cc, aws-nitro}. 
TEEs are now regarded as key to safeguarding datasets and models in ML processes~\cite{impact-of-cc, future-cc1}.
The very early work~\cite{ohrimenko2016oblivious} deploys the full ML training task inside one TEE, thus working only for small datasets and models.
Later, Chiron~\cite{hunt2018chiron} is designed for 
a confidential model hosted by the ML-as-a-Service provider to be privately trained on a single dataset owner's dataset.  It uses SGX with constrained resources and limited isolation, and does not support FL-style collaborative training across multiple mutually distrustful dataset owners.

After that, TEEs have been gradually applied to FL scenarios.
SecureFL~\cite{quoc2021secfl} 
leverages TEEs in an FL scenario to ensure confidentiality and integrity of the FL training.
PPFL~\cite{mo2021ppfl} proposes a layer-wise partitioning where only some model layers are trained inside TEEs, which cannot protect the full model confidentiality.
Papaya~\cite{huba2022papaya} uses TEEs to implement secure aggregation to support asynchronous and synchronous FL training in production.
AWS Clean Rooms ML~\cite{aws-clean-rooms}, Azure Clean Rooms~\cite{azurecleanroom} and Snowflake Data Clean Rooms~\cite{snowflake-clean-rooms} leverage TEEs to safeguard 
datasets and ML models, with the help of a look-alike model for data richness. 

More recently, Citadel~\cite{zhang2021citadel} employs the process-level TEEs and applies the zero-sum masking to protect the confidentiality of datasets and models, but without the DP guarantees. Confidential Federated Computations~\cite{eichner2024confidential} leverages TEEs and open-source verification to further enable DP for confidential FL training.  
DeTA~\cite{cheng2024deta} combines decentralized TEEs with parameter-level shuffling to limit leakage from compromised aggregation. GPU Travelling~\cite{zhao2025gputravelling}  
accelerates confidential collaborative training by dynamically assigning TEE-enabled GPUs directly to dataset owners' CVMs, bypassing slow network data streaming.  While these recent systems differ in what ML assets they protect, they \emph{all} suffer from the same problem that they cannot protect datasets when the model and training code are malicious, thus requiring them to be inspected (breaking model confidentiality).

There have also been cryptographic approaches to  enable confidential collaborative training, based on  homomorphic encryption~\cite{bost2014machine, brickell2007privacy, graepel2012ml, li2017multi, aono2017privacy, stevens2022fldp} or secure multi-party computation~\cite{nandakumar2019towards, tian2022sphinx, hesamifard2018privacy, sav2020poseidon, mohassel2017secureml, agrawal2019quotient, mohassel2018aby3, wagh2019securenn, patra2020blaze, chaudhari2019trident, koti2021tetrad, ng2023sok, wagh2020falcon, liu2024pencil} or both~\cite{choquettechoo2021capc}.  They all impose significant overheads~\cite{liu2021machine}.

\section{Conclusion}
\label{sec:conclusion}

\sys is the first collaborative training system that simultaneously protects the confidentiality of datasets, models and
training code as well as the privacy of individual user data,
even if the model and training code are malicious and uninspectable due to confidentiality needs.
It designs a chained approach of VM-level TEEs, 
a privacy barrier for improved privacy, 
a VMPL-based sandboxing to
enforce privacy barrier, 
and an integrity mechanism to protect the deployment and execution of privacy barrier and sandboxing.
\sys provides strong guarantees with modest overheads. Importantly, it paves the way to support previously-infeasible ML training workloads in a mutually distrustful environment.

\clearpage
\balance

\bibliographystyle{plain}
\bibliography{contents/main}

\clearpage
\pagenumbering{arabic}
\appendix 

\section*{Appendix}

\section{Formal Differential Privacy Analysis}
\label{sec:security_analysis_formal_dp}

\sys provides data privacy during and after training via our privacy barrier (\S\ref{sec:design-dp}) that improves upon the DP-SGD related techniques~\cite{abadi2016deep, bonawitz2017practical, kairouz2021distributed, andrew2021differentially, golatkar2022mixed}.
Accounting for DP guarantees in DP-SGD can be done via the methods described in~\cite{koskela2020, gopi2021}.
Specifically, our scheme can be seen as an adaptive composition of DP mechanisms (i.e., Gaussian mechanisms and subsampled Gaussian mechanisms). %
Using the results of~\cite{zhu2022optimal,gopi2021}, this composition can be evaluated using \textit{privacy loss random variables} (PLRVs) determined by the individual mechanisms.
The proofs for DP guarantees and the details of PLRVs can be found in Appendices~\ref{sec:dp_accounting} and \ref{sec:dp_accounting2}.

We design three additional privacy mechanisms for \sys: 1) individual-item processing, 2) dynamic gradient clipping, and 3) DP masking.
The privacy leakage of a mechanism can be assessed by how much the mechanism's output changes with a single data item, and this guides the computation of the DP privacy parameters $(\epsilon, \delta)$.

\noindent\textbf{Single-item processing (\S\ref{sec:privacy_barrier:item}).} This mechanism enforces that each data handling process handles only a single data item at a time, and ensures that the information contained in a single data item can only be transferred to a single gradient item.  This mechanism does not change the behavior of the standard DP-SGD (which assumes the gradient computation is executed faithfully but is vulnerable to gradient manipulation attacks, see \S\ref{sec:privacy_barrier:item} \& \S\ref{sec:analysis:privacy}).

\noindent\textbf{Dynamic Gradient Clipping (\S\ref{sec:privacy_barrier:clipping}).}
The histogram aggregation for dynamic gradient clipping is a Gaussian mechanism.
During aggregation, changing a single data item changes the count of each of the two histogram bins at most by 1, meaning the $L_2$-sensitivity is at most $\sqrt{2}$.
This sensitivity and noise scale $\sigma_g$ are sufficient to compute the DP guarantees of the gradient clipping mechanism,
which are combined with the DP-SGD guarantees for the final model.

\noindent\textbf{DP masking (\S\ref{sec:privacy_barrier:dpmasking}).}
DP masking does not change the final model's DP guarantees, as the privacy leakage and the noise amount are unchanged from DP-SGD.
It simply extends those guarantees even when the model owner and up to $n-1$ dataset owners collude (\S\ref{sec:analysis:privacy}).

\subsection{Background: Gaussian Mechanism and DP Accounting}\label{sec:dp_accounting}

In a DP ML training process, each gradient leaks some information about the training dataset. DP-SGD includes a privacy accounting scheme which tracks the total $(\veps, \delta)$-privacy loss  over multiple iterations of the training process and ensures that the total cumulative effect of all successive gradient updates still satisfies the desired privacy guarantees.

In the basic case of full-batch gradient training, the privacy loss of DP-SGD is that of the Gaussian mechanism, where the randomized function $\mathcal{M}$ is of the form $$
\mathcal{M}(X) = f(X) + Z,
$$
where $f \, : \, \mathcal{X}^N \rightarrow \mathbb{R}^d$and $Z \sim \mathcal{N}(0, \sigma^2 I_d)$, where $\sigma$ denotes the noise scale parameter. Denote the sensitivity of the function $f$ by 
$$
\Delta = \max_{X \simeq X'} \norm{f(X) - f(X')}_2,
$$ 
where $X \simeq X'$ denotes that $X'$ is obtained from $X$ by a change of a single data entry (similar to the neighboring datasets definition from \S\ref{sec:dpsgd}). Then, $\mathcal{M}$ is $(\veps,\delta)$-DP for $\delta(\veps)$ given by
\begin{equation} \label{eq:delta_gaussian}
    \delta(\veps) = \Phi\left( - \frac{\veps\sigma}{\Delta} + \frac{\Delta}{2\sigma} \right)
- \ee^\veps \Phi\left( - \frac{\veps\sigma}{\Delta} - \frac{\Delta}{2\sigma} \right),
\end{equation}
where $\Phi$ denotes the CDF of the standard univariate Gaussian distribution~\cite{balle2018improving}.
The bound \eqref{eq:delta_gaussian} for the Gaussian mechanism is tight, meaning that there is no
$\delta'<\delta$ such that the mechanism would be $(\veps,\delta')$-DP. When we apply a Gaussian mechanism with sensitivity $\Delta$ and noise scale $\sigma$ for $T$ times (e.g., $T$ iterations of DP-SGD will full data at each iteration), the privacy guarantee is given by Eq.~\eqref{eq:delta_gaussian} with $\Delta$ replaced by $\sqrt{T} \Delta$~\cite{sommer2019privacy}.

For more complex mechanisms such as DP-SGD with random mini-batches, analytical formulas such as \eqref{eq:delta_gaussian} do not exist, and one has to employ \it privacy accounting. \rm The so-called R\'enyi Differential Privacy (RDP)~\cite{abadi2016deep,mironov2017} was the first approach to accurately compute the privacy bounds for DP-SGD, and it is implemented in frameworks such as Opacus~\cite{opacus2021} and Tensorflow Privacy~\cite{tensorflowprivacy}. RDP has certain inherent inaccuracy and commonly still slightly overestimates the $\veps$-values. Recently, so-called numerical accounting methods~\cite{koskela2020,gopi2021,zhu2022optimal} have been proposed which give tight $\veps$-bounds for several methods including DP-SGD. These are based on so-called privacy loss random variables.

\subsection{DP Accounting Using Privacy Loss Random Variables}\label{sec:dp_accounting2}

We can generally find $(\veps,\delta)$-bounds for DP mechanisms by analysing so-called
dominating pairs of distributions:
\begin{definition}[\cite{zhu2022optimal}]
A pair of distributions $(P,Q)$ 
is a dominating pair of distributions for mechanism $\mathcal{M}$ if for all $\alpha \geq 0$,
$$
\max_{X \sim X'} H_\alpha(\mathcal{M}(X) || \mathcal{M}(X')) \leq H_\alpha(P || Q),
$$
where $H_\alpha( \cdot || \cdot)$ denotes the hockey-stick divergence, i.e., for distributions $P$ and $Q$,
$$
H_\alpha(P||Q) = \int [P(t) - \alpha \cdot Q(t)]_+ \, \dd t.
$$
\end{definition}
When evaluating the privacy guarantees of DP-SGD iterations and Gaussian mechanisms we know exactly the dominating pairs of distributions in both cases~\cite{zhu2022optimal}, and their compositions can be evaluated accurately using the methods presented in~\cite{koskela2020,gopi2021} via mature software implementations such as that of Opacus~\cite{opacus2021}.

Having dominating pairs of distributions for each individual mechanism in a composition, the hockey-stick
divergence can be transformed into a more easily computable form by using the privacy loss random variables (PLRVs).
PLRV for a pair of distributions $(P,Q)$ is defined as follows.
\begin{definition} \label{def:pld}
Let $P(t)$ and $Q(t)$ be probability density functions.
We define the PLRV $Y_{P/Q}$ as   %
\begin{equation*}%
	\begin{aligned}
	Y_{P/Q} = \log \frac{P(t)}{Q(t)}, \quad t \sim P(t),
	\end{aligned}
\end{equation*}
where $t \sim P(t)$ means that $t$ is distributed according to $P(t)$.
\end{definition}

The $\delta(\veps)$-bounds can be stated using the following representation that involves the PLRV.
\begin{thm}[\cite{gopi2021}] 
We have:
\begin{equation} \label{eq:omega_integral}
H_{\ee^\veps}(P||Q) = \mathbb{E}_{ Y_{P/Q}} \left[ 1 - \ee^{\veps-Y_{P/Q}}\right]_+,
\end{equation}
\end{thm}

Moreover, by the results of~\cite{gopi2021} and~\cite{zhu2022optimal} we have the following result for an adaptive composition of DP mechanisms.

\begin{thm}[\cite{gopi2021,zhu2022optimal}] \label{thm:composition}
In case of an adaptive composition of $T$ mechanisms $\mathcal{M}_1, \ldots, \mathcal{M}_T$, where each mechanism $\mathcal{M}_i$, $i \in [T]$, has a dominating pair of distributions $(P_i,Q_i)$, the total DP guarantees $\big(\veps,\delta(\veps) \big)$ are given by the expression
$$
\mathbb{E}_{ Y} \left[ 1 - \ee^{\veps-Y}\right]_+,
$$
where $Y = Y_1 + \ldots + Y_T$ and each $Y_i$ is a PLRV determined by the dominating pair of distributions $(P_i,Q_i)$, respectively.
\end{thm}

The DP-SGD privacy accounting can be seen as an adaptive composition of DP mechanisms (e.g., Gaussian mechanisms or subsampled Gaussian mechanisms), each with independent randomness. Therefore, these compositions can be evaluated using PLRVs determined by the individual mechanisms, i.e., using Theorem~\ref{thm:composition}.

We also get from Theorem~\ref{thm:composition} the following theorem for privacy accounting of DP-SGD combined with dynamic gradient clipping.

\begin{thm} \label{thm:dpsgd}
Consider a model training with $T$ iterations of DP-SGD with noise scale $\sigma$ and subsampling ratio $q$ (batch size divided by the total dataset size) and suppose there are in total $n_g$ noisy aggregations of the histograms for gradient clipping, each with the noise parameter $\sigma_g$. Then, for the final model, an $(\veps,\delta)$-upper bound is given by the expression
\begin{equation}  \label{eq:delta_pld_expression}
    \delta(\veps) = \mathbf{E}_{Y}[1 - e^{\veps-Y}]_+,
\end{equation}
where $[z]_+ = \max \{z,0\}$ and the random variable 
$$
Y= \sum_{i=1}^T Y_i + \sum_{i=1}^{n_g} \widetilde Y_i,
$$
where each random variable $Y_i$ is a PLRV of DP-SGD for the parameters $\sigma$ and $q$ and each random variable $\widetilde Y_i$ is a PLRV of the Gaussian mechanism for the parameter $\sigma_g$.
\end{thm}

We remark that the privacy guarantees of the expression given in Theorem~\ref{thm:dpsgd} can be readily evaluated using the implementation of the numerical accounting given in the Opacus library.

\section{Implementation Details of VMPL-Enforced Sandboxing}
\label{sec:sandbox-implementation}

This section expands on the execution workflow, isolation enforcement, and
performance optimizations in \S\ref{sec:sandbox}. Each sandbox runs the
untrusted data handling code in a fresh child process that receives the current
model and exactly one data item. The trusted service code and sandbox monitor
at VMPL0 control its inputs, permitted operations, and output.

\subsection{Proxy Kernel and Execution Lifecycle}

The proxy kernel provides execution support at VMPL1, while the sandbox monitor
at VMPL0 enforces isolation. The proxy
kernel supplies descriptor tables, fault stubs, and syscall forwarding~\cite{mei2024cabin}. It is
part of the minimal TCB and is trusted to preserve isolation among VMPL1
sandboxes. The sandbox monitor independently checks requests at VMPL0.

\noindent\textbf{Preparation and Entry.}
Following \S\ref{sec:sandbox:operation}, the trusted service code prepares the
current iteration's model and a reusable template using trusted setup code,
without processing any data items. Setup uses public shape metadata. The
service code forks a copy-on-write child process, supplies one data item, and
submits its memory regions, entry state, and policy to the sandbox monitor.
The sandbox monitor records the child's page permissions, stores a private
policy copy, and initializes a fresh Virtual Machine Save Area (VMSA) holding
the VMPL1 register state. It adopts the VMSA at VMPL1 with the SEV-SNP's
\texttt{AP\_CREATE} operation, and transfers control to the proxy kernel at
ring~0, which invokes the data handling code at ring~3.

\noindent\textbf{Completion and Cleanup.}
The proxy kernel returns the gradient output and completion status through the
sandbox monitor. The sandbox monitor releases the finished process's page
permissions, subject to shared-page reference counts, and clears its VMSA.
The trusted service code checks the output's tensor shapes, numeric types,
finite values, and total size. If execution fails, times out, or returns an
invalid output, the service code substitutes a zero vector. This does
not introduce privacy issues because, afterwards, the service
code will apply the privacy barrier, while the completion
status and diagnostic messages stay private in the service code,
as required by \S\ref{sec:sandbox:isolation}. The service code also reaps the
child process, and the guest kernel reclaims its private address space.
Changes to the child's private state never propagate back to the template or
become inputs to a later invocation.

\subsection{Page Permissions and Request Handling}

\noindent\textbf{Page Permissions.}
SEV-SNP's Reverse Map Table (RMP) constrains access to physical frames independent
of the child process's page tables. At entry, the sandbox monitor walks the
child's declared mappings and uses \texttt{RMPADJUST} to grant each frame only
the VMPL1 permissions required by its mapping. It records the granted frames
and pinned-page references for cleanup. Read-only file-backed pages can remain
shared, while writable state becomes private through copy-on-write. The
communication buffer receives only the permissions needed to exchange requests
with the sandbox monitor. At teardown, the sandbox monitor releases the child's
references and revokes permissions that are no longer needed; shared frames
retain their permissions until the last sandbox releases them.

\noindent\textbf{Requests and Channels.}
The proxy kernel places requests in the shared communication buffer and
executes a designated \texttt{CPUID} instruction. With \texttt{ReflectVC}
enabled, the resulting \#VC exception is reflected by the untrusted hypervisor
to VMPL0. The sandbox monitor checks each request and its relevant arguments
against the private policy copy fixed before execution. 
The proxy kernel’s copy provides an additional layer
of defense and serves as an execution aid for performance.

The policy denies socket, file access and inter-process communication,
except for communication with the sandbox monitor and the permitted private
diagnostic pipes. In particular, it does not permit \texttt{openat},
pathname-based metadata calls, shared memory, or message queues. Argument
checks restrict writes to private diagnostic pipes and prevent memory-management
operations from overlapping the communication buffer. The sandbox monitor
emulates calls that would alter the VMPL0 executor's task state; for example,
it answers signal-handler changes without installing the untrusted handler.
A watchdog on another vCPU injects an Non-Maskable Interrupt (NMI) when a sandbox exceeds its deadline.
These checks remain authoritative regardless of the requester's privilege
level within VMPL1.

\subsection{Sandboxing Overhead Optimizations}
\label{sec:sandbox-optimizations}

The optimizations in \S\ref{sec:sandbox:optimizations} reduce the repeated runtime
setup and memory management work while retaining a fresh sandbox for each data
item. The implementation makes the following choices.

\noindent\textbf{Template Reuse.}
The persistent template retains the Python interpreter, PyTorch libraries,
current model and dispatch state across iterations. Before updating the model,
the service code waits for active sandboxes to finish. Each subsequent child
therefore starts with the same model for that iteration but receives its own
data item, private writable state, VMSA, and page-permission accounting.
Template reuse avoids the repeated imports and model construction without carrying
private state between invocation.

\noindent\textbf{Shared Page Permission.}
The sandbox monitor tracks shared frames with per-frame reference
counts, and partitions permission-management locks by physical frame. Thus,
teardown preserves pages still used by concurrent sandboxes, and unrelated
permission updates do not serialize on a global lock. Preserving each mapping's
existing permissions allows immutable file-backed pages to remain shared.

\noindent\textbf{Copy-on-Write.}
The service code uses the process's resident memory footprint to decide when
to create private copies of writable pages. For CIFAR10-CNN6, we use lazy copy-on-write:
pages remain shared until the child process first writes to them. The resulting write fault
invokes the sandbox monitor to arrange a private copy and update its page
permissions. This avoids copying pages that the child process never modifies.
For AGNEWS-Roberta, the larger working set makes repeated VMPL1 exits, page
copying, and permission updates costly under concurrent execution. This
workload therefore uses template reuse but creates private copies of writable
pages before entering VMPL1, avoiding write faults to copy those pages during
sandbox execution. Figure~\ref{fig:sys-eval-sandbox-ablation} compares these
choices and their combined effect on data handling time.

\section{Evaluation Setup}
\label{sec:evaluation-setup}
\textbf{Configuration.} To evaluate the system performance with TEEs, we deploy
AMD SEV-SNP VMs on servers equipped with a 96-core AMD EPYC 9T95 (Turin)
processor (192 vCPUs) operating at 3.5\,GHz. The testbed uses a patched host
kernel, QEMU, and guest kernel. Each VM has 33 \mbox{vCPUs} and 64\,GB of memory,
leaving one vCPU available to the sandbox monitor while 32 sandboxes execute
concurrently. By default, these experiments use PyTorch 2.9.1 with four data handling components, 
one model updating component and one admin component, and with global batch sizes ranging from 128 to 512 
and per-worker local batch sizes ranging from 32 to 128.

\noindent\textbf{Workloads.} We evaluate the following two models:
\begin{itemize} 

\item \textbf{CNN6}: A 6-layer CNN model~\cite{tian2022sphinx},  
which is often used or similar to the models used to evaluate many systems for privacy-preserving training~\cite{mo2021ppfl, mo2020darknetz, liu2024pencil, patra2020blaze}.
This model is trained over the CIFAR-10 dataset~\cite{krizhevsky2009learning}, which contains 60K RGB images (50K for training, 10K for testing) of objects in ten categories (e.g., airplane, automobile, bird), each with 32×32 pixels.

\item \textbf{Roberta}: A pre-trained language model with 12 transformer layers~\cite{roberta-base}, which is large in the context of privacy-preserving collaborative training -- larger than or at the same level as some of the largest models evaluated in this context (e.g., the Wide-ResNet-40 and NFNet-F3 models~\cite{de2022unlocking}). 
It also matches the scale of many realistic models used in various industrial domains~\cite{Uni-Mol2,Chronos-2,TimesFM,Kronos-model}.
This model is trained over the AG-News dataset~\cite{agnews}, which is constructed by assembling titles and description fields of articles from the 4 largest classes (“World”, “Sports”, “Business”, “Sci/Tech”) of AG’s Corpus and has 120K training samples and 7.6K testing samples.
\end{itemize}

For each training session, we launch one admin component, one model updating component, and a varying number of data handling components that handle different numbers of dataset owners and different batch sizes.
Models are trained  with 1000 iterations and a learning rate of 0.01.

\end{document}